\documentclass[12pt]{article}
\usepackage{amsmath}
\usepackage{amsfonts,amsthm}
\usepackage{latexsym}
\usepackage{amscd}
\usepackage{amssymb}
\usepackage{graphicx}
\usepackage{float}
\usepackage[applemac]{inputenc}
\usepackage{dsfont}
\usepackage{multirow}
\usepackage{subfigure}
\usepackage{dcolumn}
\usepackage{bbm}
\usepackage{rccol}
\usepackage{url}
\usepackage{enumerate}
\usepackage{scalefnt}
\usepackage{booktabs}
\usepackage{lineno}

\usepackage[bookmarks=false,colorlinks]{hyperref}

\usepackage{textcomp}
\usepackage{array}
\usepackage{supertabular}
\usepackage{hhline}

\makeatletter
\newcommand\arraybslash{\let\\\@arraycr}
\makeatother
\newcolumntype{+}{>{\global\let\currentrowstyle\relax}}
\newcolumntype{^}{>{\currentrowstyle}}

\newlength{\bracewidth}

\def\fudge{\mathchoice{}{}{\mkern.5mu}{\mkern.8mu}}
\def\bbc#1#2{{\rm \mkern#2mu\vbar\mkern-#2mu#1}}
\def\bbb#1{{\rm I\mkern-3.5mu #1}}
\def\bba#1#2{{\rm #1\mkern-#2mu\fudge #1}}
\def\bb#1{{\count4=`#1 \advance\count4by-64 \ifcase\count4\or\bba A{11.5}\or
   \bbb B\or\bbc C{5}\or\bbb D\or\bbb E\or\bbb F \or\bbc G{5}\or\bbb H\or
   \bbb I\or\bbc J{3}\or\bbb K\or\bbb L \or\bbb M\or\bbb N\or\bbc O{5} \or
   \bbb P\or\bbc Q{5}\or\bbb R\or\bbc S{4.2}\or\bba T{10.5}\or\bbc U{5}\or
   \bba V{12}\or\bba W{16.5}\or\bba X{11}\or\bba Y{11.7}\or\bba Z{7.5}\fi}}

\restylefloat{float}

\newtheorem{theorem}{Theorem}[section]

\newtheorem{remark}{Remark}[section]
\newtheorem{lemma}{Lemma}[section]

\title{Vector-borne diseases models with residence times - a Lagrangian perspective}

\author{Derdei Bichara$^{a, b}$ and Carlos Castillo-Chavez$^{c}$ \\
$^{a}$ Department of Mathematics, California State University, Fullerton,\\
$^{b}$ Center for Computational and Applied Mathematics, \\
800 N. State College Blvd, Fullerton, CA 92831.\\
$^{c}$ Simon A. Levin  Mathematical, Computational and Modeling Science Center, \\
Arizona State University, Tempe, AZ 85287\\
E-mails: dbichara@fullerton.edu \& ccchavez@asu.edu\\
}
\date{ }

\begin{document}

\maketitle
\abstract{A multi-patch  and multi-group modeling framework describing the dynamics of a class of diseases driven by the interactions between vectors and  hosts structured by groups is formulated. Hosts' dispersal  is modeled in terms of patch-residence times with the nonlinear dynamics taking into account the \textit{effective} patch-host size.  
The residence times basic reproduction number $\mathcal R_0$ is computed and shown to depend on the relative environmental risk of infection. The model is robust, that is, the disease free equilibrium is globally asymptotically stable (GAS) if $\mathcal R_0\leq1$ and a unique interior endemic equilibrium is shown to exist that is GAS whenever $\mathcal R_0>1$ whenever the configuration of host-vector interactions is irreducible. The effects of \textit{patchiness} and \textit{groupness}, a measure of host-vector heterogeneous structure, on the basic reproduction number $\mathcal R_0$, are explored. Numerical simulations are carried out to highlight the effects of residence times on disease prevalence.}\\


\section*{Introduction}
Vector-borne diseases, a major public health problem around the world, are responsible for over one million death and hundreds of millions cases each year \cite{world2014global,tolle2009mosquito} and so diminishing their impact is a worldwide priority.  Travel, climate change and trade have significantly altered vector-borne diseases dynamics \cite{castillo2015beyond,elbers2015mosquitoes,kuno1995review,perrings2014merging,reiter2014climate}. Sir Ronald Ross \cite{Ross1911} was the first to model a vector borne disease dynamics. Ross's paper  \cite{Ross1911} and follow up work \cite{Ross1916,RossHud1917,RossHuds17II} laid the foundation of what is known today as the field of mathematical or theoretical epidemiology. There is an extensive   literature associated with the  the study of vector-host interactions in the context of human diseases  (\cite{Arino_Bow07,aron1982population,bailey1982biomathematics,Chit_08,MR2272613,Chit_Hym_Cus_08,Dietz75,MR83d:92066,MR83c:92058,dumont2010vector,dumont2008temporal,hasibeder1988population,macdonald1956epidemiological,macdonald1957epidemiology,NgwaBMB06,NgwaMCM00} and the references therein). 
Sparse theoretical results exist on the role of geographical heterogeneity on the spread of vector-borne diseases, mostly via 
metapopulation models \cite{adams2009man,arino2012metapopulation,MBS6900,cosner2009effects,gao2012multipatch,torres1997models,rodriguez2001models,smith2004risk,xiao2014transmission}, that assume that the movement of host is ``permanent"; this approach has been referred as Eulerian \cite{grunbaum1994translating,okubo2013diffusion,okubo1980diffusion}. A Lagrangian perspective considers the movement of individuals across patches in a framework where the hosts' origin or identity are never lost. This approach, useful in the study of the role of movement of individuals in highly connected settings albeit it has received limited attention \cite{cosner2009effects,dye1986population,iggidr2014dynamics,rodriguez2001models,Ruktanonchai2016}.

The concept of Langragian and Eulerian approaches were implemented by Okubo et \textit{al.} \cite{okubo2013diffusion,okubo1980diffusion} in modeling the diffusion and aggregation of animal populations in ecology. This nomenclature has been used in the context of epidemic models by Cosner et \textit{al.} \cite{cosner2009effects}. The use of a Lagrangian approach in the study of the dynamics and control of vector-borne diseases has also been explored in \cite{dye1986population,hasibeder1988population} prior this work. Specifically, Dye and Hasibeder  \cite{dye1986population,hasibeder1988population} considered the study of vector-born dynamics via $SIS-SI$ type host-vector models in the context of $n$ patch systems. Rodriguez and Torres-Sorando \cite{rodriguez2001models} used a Lagrangian perspective via the  incorporation of short-time visitations to multiple patches, also in the context of vector borne disease. In \cite{Ruktanonchai2016}, authors also considered a patchy Ross-Macdonald model and derived patch specific basic reproduction number in order to identify which patch is a source or a sink. 
More recently, Iggidr et \textit{al.} \cite{iggidr2014dynamics}  introduced a general $SIR-SI$ multi group deriving necessary and sufficient conditions for the existence of a sharp threshold \cite{iggidr2014dynamics}. Their \cite{iggidr2014dynamics} abstract setting  did not incorporate residence times explicitly, albeit their general infection terms technically may allow for their inclusion. The study in  \cite{iggidr2014dynamics} and related papers, with the exception of \cite{hasibeder1988population,dye1986population}, assume that  hosts and vectors are residents or members of particular patch or group. Our framework can handle multiple levels of organization including the host's age or socio-economic structure (see \cite{martens2000malaria,thiboutot2010chikungunya} for the age factors and \cite{biritwum2000incidence,koram1995socio,onwujekwe2008improving} for the socio-economics' role). Since vector transmission is often determined by the vectors' place of residence, it is often useful to decouple the host's structure from that of vectors' population whenever possible.

In this paper, we consider a vector-host model where the host population is structured by groups/classes that interact with non-mobile vectors living in multiple patches/environments. The hosts' groups may be defined by socioeconomic background, gender, or age. The vectors' patches represent the vectors' ``space", which include schools, farms, workplaces etc. Hosts, in general, will distribute their time in a multitude of vectors' places of residence (patches).
In our setup, we assume that the spatial scale under consideration is such that ignoring vector mobility across patches is acceptable.  There are evidences that such an assumption is reasonable, for example, Dengue and Chikungunya's urban vectors \textit{Aedes aegypti} rarely travel more than a few tens of meters during their lifespan \cite{adams2009man,WHODengueMosquitoes}; the mainly rural but urban adapted vector \textit{Aedes albopictus} have maximum dispersal of 400-600 meters \cite{honorio2003dispersal,niebylski1994dispersal}; according to \cite{bonnet1946dispersal,niebylski1994dispersal}, the vectors \textit{Aedes albopictus} are unlikely to travel long distance due to wind speed variability, in fact, they exhibit a tendency to fly closer to the ground, desisting  to fly during heavy winds; the adult \textit{Anopheles}( vector of malaria) does not fly more than 2 kilometers \cite{sturchler2001vector}; and, \textit{Anopheles gambiae}'s (the main malaria vector in Africa) maximal flight distance is 10 kilometer \cite{kaufmann2004flight}.  In short, the spread of vector-borne diseases, in many instances, is primarily due to hosts' dispersal. Therefore, it is assumed here as in \cite{MBS6900,xiao2014transmission} that vectors do not abandon their geographical environment or patch. There are alternative modes of mosquitoes dispersal like those generated by trade, including the used-tires' trade \cite{novak1995north,roiz2007survey}.

The host population is structured into $n$ groups with dispersal modeled via the residence times matrix $\mathbb P=(p_{ij})_{\substack{
      1 \leq i \leq n, \\
      1 \leq j \leq m}}$, where $p_{ij}$ denotes the proportion of time that a host member of Groups $i$ spends in Patch $j$.  The use of this approach impacts the temporal dynamics of the \textit{effective} host population size in each patch. Host \textit{effective} population size \textit{per} patch, that is the number of hosts of each group  at time $t$ in Patch $j$, $j=1,2,\dots,m$; is computed using the entries of the matrix $\mathbb P$ as weights. The density of \textit{effective} infected host per patch account for both \textit{effective} population and \textit{effective} infected population size  in each patch.

      The host \textit{effective} population size has not been incorporated in the literature using a Lagrangian approach in the contet of vector-borne diseases before \cite{cosner2009effects,rodriguez2001models} (but see \cite{castillo2003epidemic}). Our formulation generalize the case where vectors and hosts are defined  by jointly inhabited patches \cite{cosner2009effects,iggidr2014dynamics,rodriguez2001models}. We prove that the disease free equilibrium is GAS if $\mathcal R_0\leq1$ and that a unique endemic equilibrium exists and is GAS if $\mathcal R_0>1$ whenever the multi-patch, multi-group system is irreducible. This approach has been used in the study of a general $SIS$ model in the context of communicable diseases (\cite{bichara2015sis}).

The paper is organized as follow. Section 1 is devoted to the derivation and basic properties of the model; Section \ref{sec:R0GAS} deals with the stability analysis of the disease free equilibrium (DFE) and the endemic equilibrium. Section \ref{sec:HeteroR0}, highlights the role of heterogeneity in term of patch and group variability on the basic reproduction number; Section \ref{sec:Simulations} highlights tour results in the context of 2 groups, 2 patches and 2 groups and 3 patches via  simulations. Section \ref{sec:Conclusion} collects our conclusions and thoughts on the usefulness of this approach and list possible extensions. 
\section{Derivation of the model}\label{sec:Derivation}

We consider the dynamics of  human-vector interactions within a population composed of $n$ social groups and $m$ environments or patches. We denote by $N_{h,i}$  the host population in social group $i$, $i=1,\dots,i$, and $N_{v,j}$ vector population in Patch $j$, $j=1,\dots,m$. The susceptible and infected host populations in group $i$, $i=1,\dots,n$ , at time $t$, are denoted by $S_{h,i}(t)$ and $I_{h,i}(t)$ respectively.  It is assumed that the total host population in each group is constant, that is $N_{h,i}=S_{h,i}(t)+I_{h,i}(t)$; that the disease in the host is captured by an $SIS$ epidemic model while the vectors' dynamics follows an $SI$ framework. The vector population in each patch is composed by $S_{v,j}$ and $I_{v,j}$, the susceptible and infected vector populations in Patch $j$, $j=1,\dots,m$, respectively.  

The entries of the residence times matrix $\mathbb P$ denote the proportion of time that individuals of different groups spend in each patches; specifically $p_{ij}$ represents the proportion of  time that members of group $i$ spend in Patch $j$ ($p_{ij}\geq0$ for all $j$ and $\sum_{j=1}^mp_{ij}=1$ for all $i$). The susceptible individuals of group $i$ ($S_{h,i}$) are generated through birth at the per-capita rate $\mu_i$ and they recover from infection at the per-capita rate $\gamma_i$. It is assumed that all offsprings are susceptible and that the disease does not confer immunity. The birth of susceptible individuals in group $i$ is compensated by deaths, maintaining constant host population size in each group. The host population is at risk of infection in every patches from its interaction with local infected vectors ($I_{v,j},\; j=1,\dots, m$). Hence, the dynamics of the the susceptible host of group $i$, for $i=1,\dots,n$, is given by:
$$
\dot S_{h,i}=\mu_iN_{h,i}+\gamma_iI_{h,i}-\sum_{j=1}^{m}b_j(N_h,N_{v,j}) \beta_{v,h}p_{ij}S_{h,i}\frac{I_{v,j}}{N_{v,j}}
-\mu_i S_{h,i}$$ 
where $b_j(N_h,N_{v,j}) $ is the number of mosquito bites per human per unit of time \cite{CCCContacts,MR2272613,forouzannia2014mathematical,garba2008backward}  in Patch $j$. $b_j(N_h,N_{v,j})$ is assumed to be a function of the number of host in Patch $j$; a population that includes visitors from other patches.

The dynamics of infected hosts of group $i$, $i=1,\dots,n$, is  modeled as follows

\begin{equation}\label{EqIh}
\dot I_{h,i}=\sum_{j=1}^{m}b_j(N_h,N_{v,j}) \beta_{vh}p_{ij}S_{h,i}\frac{I_{v,j}}{N_{v,j}}
-(\mu_i+\gamma_i) I_{h,i}
\end{equation}

The susceptible vectors in Patch $j$ are replenished via constant recruitment $\Lambda_{v,j}$, subject to death at the per-capita rate $\mu_v$  and removed (through harvesting and spraying) at the per-capita rate $\delta_j$. We suppose that the natural per-capita vectors' death rates are the same in all patches. Though, the vectors do not move across patches, the susceptible mosquitoes in Patch $j$ ($S_{v,j}$) may, of course, be infected by infected hosts of any group while visiting Patch $j$. The \textit{effective} proportion of infected individuals in Patch $j$ is therefore given by

$$\frac{ \sum_{i=1}^{n}p_{ij}I_{h,i}}{ \sum_{k=1}^{n}p_{kj}N_{h,k}}$$
Hence, the dynamics of susceptible vector in Patch $j$, $j=1,\dots,m$ in patch $j$ at time $t$ is given by

$$\dot S_{v,j}=\Lambda_{v,i}-a_j\beta_{hv}S_{v,j}\frac{ \sum_{i=1}^{n}p_{ij}I_{h,i}}{ \sum_{k=1}^{n}p_{kj}N_{h,k}}-(\mu_v+\delta_j)S_{v,j}$$ 
 
where $a_j$ is the number of bites per mosquito per unit of time in Patch $j$, assumed to be constant.%; that is, it is assumed that there are enough hosts for the vectors to feed on at any time. 

The dynamics of infected vectors in Patch $j$ is given by 

\begin{equation}\label{EqIv}
\dot I_{v,j}=a_j\beta_{hv}S_{v,j}\frac{ \sum_{i=1}^{n}p_{ij}I_{h,i}}{ \sum_{k=1}^{n}p_{kj}N_{h,k}}-(\mu_v+\delta_j)I_{v,j}
\end{equation} 

We know that the total number of bites by mosquitoes ($a_jN_{v,j}$ in Patch $j$) should equal the total number of bites received by humans ($b_j(N_h,N_v)\sum_{k=1}^np_{kj}N_{h,k}$) \cite{Arino_Bow07,CCCContacts,NgwaMCM00}. In our case, this conservation of contact rates should be satisfied in each patch. Hence, for Patch $j$, we have:
$$a_jN_{v,j}=b_j(N_h,N_{v,j})\sum_{k=1}^np_{kj}N_{h,k}.$$
This implies that:
\begin{equation}\label{EqCons}b_j(N_h,N_{v,j})=\frac{a_jN_{v,j}}{\sum_{k=1}^np_{kj}N_{h,k}}
\end{equation}
Hence, the disease dynamics for $n$ host groups interacting in $m$ different environments subjected to resident vectors is completely described by the following system:
\begin{equation} \label{PatchGenComplete}
\left\{\begin{array}{llll}
\dot S_{h,i}=\mu_iN_{h,i}+\gamma_iI_{h,i}-\beta_{vh}S_{h,i}\sum_{j=1}^{m} a_jp_{ij}\frac{I_{v,j}}{\sum_{k=1}^n p_{kj}N_{h,k}}
-\mu_i S_{h,i}\\\\
\dot I_{h,i}=\beta_{vh}S_{h,i}\sum_{j=1}^{m} a_jp_{ij}\frac{I_{v,j}}{\sum_{k=1}^n p_{kj}N_{h,k}}-(\mu_i+\gamma_i) I_{h,i}\\\\
\dot S_{v,j}=\Lambda_{v,j}-a_{j}\beta_{hv}S_{v,j}\frac{ \sum_{i=1}^{n}p_{ij}I_{h,i}}{ \sum_{k=1}^{n}p_{kj}N_{h,k}}-(\mu_v+\delta_j)S_{v,j}\\\\
\dot{I}_{v,j}=\beta_{hv}S_{v,j}\frac{ \sum_{{i}=1}^{n}p_{ij}I_{h,i}}{ \sum_{k=1}^{n}p_{kj}N_{h,k}}-(\mu_v+\delta_{j})I_{v,j},
\end{array}\right.
\end{equation}

with $i=1,\dots, n$ and $j=1,\dots, m$. The parameters used in Model (\ref{PatchGenComplete}) are defined in Table \ref{tab:Param} and the flow diagram of the model is provided in Fig \ref{fig:Flow}.

\begin{figure}[ht!]
\begin{center}
\includegraphics[scale=1]{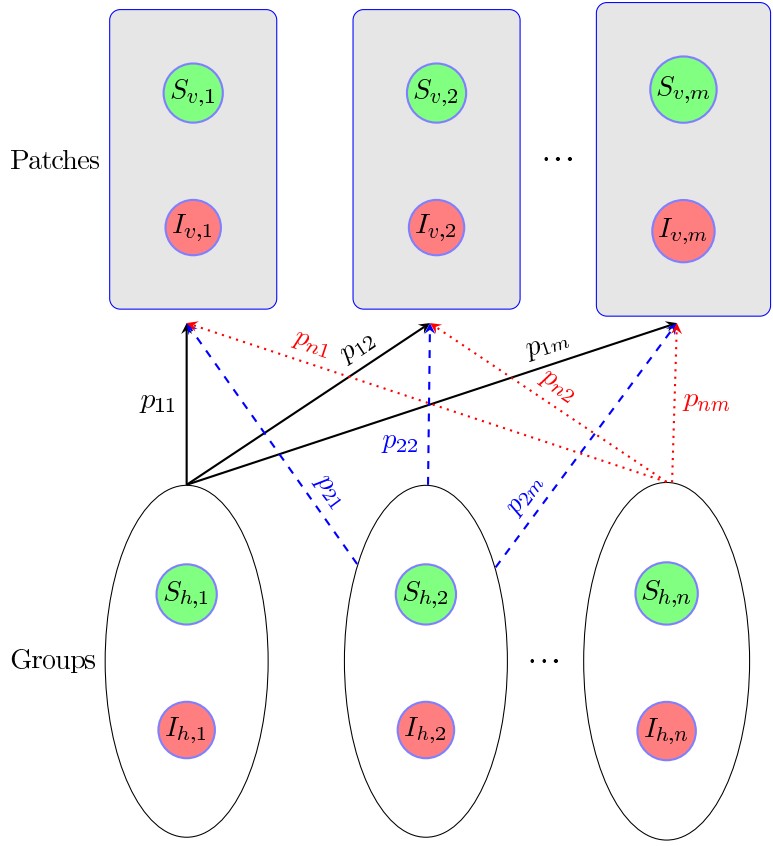}
\caption{Flow diagram of the model}
\label{fig:Flow} 
\end{center}
\end{figure}

\begin{table}[h!]
  \begin{center}
    \caption{Description of the parameters used in System (\ref{PatchGenComplete}).}
    \label{tab:Param}
    \begin{tabular}{cc}
      \toprule
      Parameters & Description \\
      \midrule
$\beta_{vh}$ & Infectiousness of human to mosquitoes per biting\\
$\beta_{hv}$ & Infectiousness of mosquitoes to humans per biting\\
$a_{j}$ &  Biting rate in Patch $j$ \\
$\mu_i$ & \textit{Per capita} humans' birth and death rate for Group $i$\\
$\gamma_i$  & \textit{Per capita} Recovery rate  for Group $i$\\
$p_{ij}$ &  Proportion of time individual in Group $i$ spend in Patch $j$\\
$\mu_{v}$  & \textit{Per capita} natural death rate of mosquitoes \\
$\delta_{j}$  & \textit{Per capita} death rate of control of mosquitoes in Patch $j$.\\
      \bottomrule
    \end{tabular}
  \end{center}
\end{table}

The total host population is constant and since total vector population dynamics are given by $$\dot N_{v,j}=\Lambda_{v,j}-(\mu_v+\delta_{j})N_{v,j},$$
we can deduce that $$\limsup_{t\to+\infty}N_{v,j}=\frac{\Lambda_{v,j}}{\mu_v+\delta_{j}}:= \bar N_{v,j} $$
And so, the vector population is asymptotically constant in each patch. The use of the asymptotic theory on triangular systems \cite{CasThi95,0478.93044}, applied to System (\ref{PatchGenComplete}), leads to the following equivalent autonomous system:
\begin{equation} \label{PatchGenF}
\left\{\begin{array}{llll}
\dot I_{h,i}=\beta_{vh}(N_{h,i}-I_{h,i})\sum_{j=1}^{m} a_jp_{ij}\frac{I_{v,j}}{\sum_{k=1}^n p_{kj}N_{h,k}}-(\mu_i+\gamma_i) I_{h,i}, \quad \forall i=1,2.\dots,n.\\\\
\dot{I}_{v,j}=a_j\beta_{hv}(\bar N_{v,j}-I_{v,j})\frac{ \sum_{i=1}^{n}p_{ij}I_{h,i}}{ \sum_{k=1}^{n}p_{kj}N_{h,k}}-(\mu_v+\delta_j)I_{v,j}, \quad \forall j=1,2,\dots,m.
\end{array}\right.
\end{equation}
Here, we have that human risk in Patch $j$ is defined by $a_j\beta_{vh}$ and so, Patch $j$ is riskier than Patch $l$ whenever $a_j>a_l$. The modeling framework is quite flexible. For example, the case $n=m$ covers the interactions between $n$ host groups (or classes) in  $n$ patches or the case when hosts and vectors are co-residents. The case $p_{ij}=0$ for all $j\neq i$ or $p_{ii}=1$ when $n=m$ leads to a collection of isolated classical Ross-Macdonal models. For $n\neq m$, hosts are structured in groups like children, farmers, retired people and vectors are distributed in patches or meeting places like home, school, farm, work place or mosquito breeding site, etc.  People from the considered groups visit patches and spend certain amount of time and get possibly infected whereby.

Dye and Hasibeder  \cite{dye1986population,hasibeder1988population} models involve $n$ host and $m$ vector patches under the assumption that only the vectors move.  The models in \cite{cosner2009effects,rodriguez2001models,Ruktanonchai2016},  have incorporated residence times explicitly but  
 their modeling does not account for  the \textit{effective} patch population size. In fact, in\cite{cosner2009effects,rodriguez2001models}, the pattern of movement between patches does not produce any ``net" change on the total population per patch at any given time. For example, it is assumed (in  \cite{rodriguez2001models}), that the total population of patch $j$ is $N/k$ where $N$ the overall human population and $k$ the number of patches. In \cite{cosner2009effects,Ruktanonchai2016}, the total population in each patch $j$ is $N_j$ (or $H_j$ in their notations) regardless of the movement of individuals between patches. Similar remarks hold for the Dengue's two-patch model in \cite{lee2015role}. In our case,  the host population in each patch is the sum of  visiting individuals of different groups \textit{weighted by the proportion of time} they spend in each patch. This means that, at time $t$, at any given Patch $j$, the host population is $p_{1j}N_{h,1}+p_{2j}N_{h,2}+\dots+p_{nj}N_{h,n}$, the \textit{effective} population size of Patch $j$. Moreover, this approach is well suited for better intervention strategies through the knowledge of $a_j$, $j=1,\dots,m$ or the residence time matrix $\mathbb P=(p_{ij})_{\substack{
      1 \leq i \leq n, \\
      1 \leq j \leq m}}$. For example, if a particular host group is more affected by the disease in consideration, that may lead to the patch within which the infection had occurred, that is the patch ``source" of infection. That could help steer control measures such as DDT in the ``infectious" patch or social distancing the ``infected" group to mitigate the disease burden.

System (\ref{PatchGenF}), could be written in a compact form as follows
\begin{equation} \label{PatchMatrices}
\left\{\begin{array}{llll}
\dot I_{h}=\beta_{vh}\textrm{diag}(N_{h}-I_{h})\mathbb P\textrm{diag}(a)\textrm{diag}(\mathbb P^tN_h)^{-1}I_v-\textrm{diag}(\mu+\gamma) I_{h}\\\\
\dot{I}_{v}=\beta_{hv}\textrm{diag}(a)\textrm{diag}(N_{v}-I_{v})\textrm{diag}(\mathbb P^tN_h)^{-1}\mathbb P^tI_h-\textrm{diag}(\mu_v+\delta)I_{v}
\end{array}\right. 
\end{equation}
where $I_{h}=[I_{h,1},I_{h,2},\dots,I_{h,n}]^t$, $I_{v}=[I_{v,1},I_{v,2},\dots,I_{v,m}]^t$, $N_{h}=[N_{h,1},N_{h,2},\dots,N_{h,n}]^t$, $\bar N_{v}=[\bar  N_{v,1},\bar  N_{v,2},\dots,\bar  N_{v,m}]^t$, $\delta=[\delta_1,\delta_2,\dots,\delta_m]^t$, $a=[a_1,a_2,\dots,a_m]^t$ and $\mu=[\mu_1,\mu_2,\dots,\mu_n]^t$.

We end this section by showing that the solutions of Model (\ref{PatchGenF}) are positive and bounded, or in other words, that the model is biologically grounded.

\begin{lemma}
The  region defined by 
\[\Omega=\left \{ (I_{h}, I_{v})\in\mathbb R^{n+m}_+\; \mbox{ \LARGE $\mid$  } \;I_{h}\leq N_{h},\; \; I_{v}\leq \bar N_{v} \right\} \] 
is a  compact attracting  positively invariant set for System (\ref{PatchMatrices}) .

\end{lemma}
\begin{proof}
The set $\Omega$, a subset $\mathbb R^{n+m}$, is clearly closed and bounded and hence a compact. The right-hand side of System (\ref{PatchMatrices}) could be written as $A(I_{h},I_v)(I_{h},I_v)^t$ where

$$A(I_{h},I_v)=\left(\begin{array}{cc}
-\textrm{diag}(\mu+\gamma) & \beta_{vh}\textrm{diag}(N_{h}-I_{h})\mathbb P\textrm{diag}(a)\textrm{diag}(\mathbb P^tN_h)^{-1}\\
 \beta_{hv}\textrm{diag}(a)\textrm{diag}(\bar N_{v}-I_{v})\textrm{diag}(\mathbb P^tN_h)^{-1}\mathbb P^t & -\textrm{diag}(\mu_v+\delta) 
\end{array}\right)$$
Since, $I_n\leq N_h$ and $I_v\leq \bar N_v$, the matrix $A(I_{h},I_v)$ is Metzler. Hence, the positive orthant $\mathbb R^{n+m}_+$ is invariant. At $I_h=N_{h}$, we have
$\dot I_{h}=-\textrm{diag}(\mu+\gamma) I_{h}\leq0$. Similarly, at $I_v=\bar N_{v}$, we have $\dot I_v=-\textrm{diag}(\mu_v+\delta)I_{v}\leq0$. Hence, the vector field of (\ref{PatchMatrices}) is pointed inward from the faces of $\Omega$.
\end{proof}
%%%%
%%%%%
%%%%%%
\section{Equilibria and Global Stability}\label{sec:R0GAS}
In the absence of infected vectors in all patches, Model (\ref{PatchMatrices}) supports a unique, disease free equilibrium (DFE), given by $E_0=\textbf{0}_{\mathbb R^{n+m}}$. The basic reproduction number, defined as the average number of secondary cases produced of by an infected individual during its  lifetime, is computed using the next generation method \cite{MR1057044,VddWat02}. The right hand side of (\ref{PatchMatrices}) could be written as $\mathcal F+\mathcal V$ where 
$$\mathcal F(I_h,I_v)=\left(\begin{array}{c}
\beta_{vh}\textrm{diag}(N_{h})\mathbb P\textrm{diag}(a)\textrm{diag}(\mathbb P^tN_h)^{-1}I_v\\
 \beta_{hv}\textrm{diag}(a)\textrm{diag}(N_{v})\textrm{diag}(\mathbb P^tN_h)^{-1}\mathbb P^t I_h
\end{array}\right)\quad \textrm{and}\quad \mathcal V(I_h,I_v)=\left(\begin{array}{c}
-\textrm{diag}(\mu+\gamma)I_h \\
 -\textrm{diag}(\mu_v+\delta) I_v
\end{array}\right) $$

Let $F=D\mathcal F(I_h,I_v)$ and $V=D\mathcal V(I_h,I_v)$ evaluated at the DFE. We obtain:

$$F=\left(\begin{array}{cc}
0 & \beta_{vh}\textrm{diag}(N_{h})\mathbb P\textrm{diag}(a)\textrm{diag}(\mathbb P^tN_h)^{-1}\\
 \beta_{hv}\textrm{diag}(a)\textrm{diag}(N_{v})\textrm{diag}(\mathbb P^tN_h)^{-1}\mathbb P^t & 0 
\end{array}\right)$$
and $$V=\left(\begin{array}{cc}
-V_h &0\\
0 & -V_v 
\end{array}\right)$$
where $V_h=\textrm{diag}(\mu+\gamma)$ and $V_v=\textrm{diag}(\mu_v+\delta)$.\\
The basic reproduction number is the spectral radius of the next generation matrix:

$$-FV^{-1}=\left(\begin{array}{cccc}
\textbf{0}_{n,n} & M_{vh}\\
M_{hv} & \textbf{0}_{m,m}
\end{array}\right)$$

where
$$M_{hv}=\beta_{hv}\textrm{diag}(a)\textrm{diag}(\mathbb P^t N_h)^{-1}\textrm{diag}(N_v)\mathbb P^tV_h^{-1}$$ and $$M_{vh}=\beta_{vh}\textrm{diag}(N_h)\mathbb P\textrm{diag}(\mathbb P^tN_h)^{-1}\textrm{diag}(a)V_v^{-1}$$

Notice also that since $(-FV)^2=\textrm{diag}(M_{vh},M_{hv})$, we can deduce that the basic reproduction number is $\mathcal R_0^2=\rho(M_{vh}M_{hv})$.

More precisely, we have that
$$M_{hv}=\left(\begin{array}{cccc}
\frac{a_1\beta_{hv}p_{11}N_{v,1}}{\sum_{i=1}^np_{i1}N_{h,i}(\mu_1+\gamma_1)} &\frac{ a_1\beta_{hv}p_{21}N_{v,1}}{\sum_{i=1}^np_{i1}N_{h,i}(\mu_2+\gamma_2)} &\cdots & \frac{ a_1\beta_{hv}p_{n1}N_{v,1}}{\sum_{i=1}^np_{i1}N_{h,i}(\mu_n+\gamma_n)}\\
\frac{a_2\beta_{hv}p_{12}N_{v,2}}{\sum_{i=1}^np_{i2}N_{h,i}(\mu_1+\gamma_1)} &\frac{ a_2\beta_{hv}p_{22}N_{v,2}}{\sum_{i=1}^np_{i2}N_{h,i}(\mu_2+\gamma_2)} &\cdots &  \frac{ a_2\beta_{hv}p_{n2}N_{v,2}}{\sum_{i=1}^np_{i2}N_{h,i}(\mu_n+\gamma_n)}\\
\vdots & \vdots& \cdots & \vdots\\
\frac{a_m\beta_{hv}p_{1m}N_{v,m}}{\sum_{i=1}^np_{im}N_{h,i}(\mu_1+\gamma_1)} &  \frac{a_m\beta_{hv}p_{2m}N_{v,m}}{\sum_{i=1}^np_{im}N_{h,i}(\mu_2+\gamma_2)} & \cdots &  \frac{ a_m\beta_{hv}p_{nm}N_{v,m}}{\sum_{i=1}^np_{im}N_{h,i}(\mu_n+\gamma_n)}
\end{array}\right)$$
and 
$$M_{vh}=\left(\begin{array}{cccc}
\frac{a_1\beta_{vh}p_{11}N_{h,1}}{\sum_{i=1}^np_{i1}N_{h,i}(\mu_v+\delta_1)} &\frac{ a_2\beta_{vh}p_{12}N_{h,1}}{\sum_{i=1}^np_{i2}N_{h,i}(\mu_v+\delta_2)} &\cdots & \frac{ a_m\beta_{hv}p_{1m}N_{h,1}}{\sum_{i=1}^np_{im}N_{h,i}(\mu_v+\delta_m)}\\
\frac{a_1\beta_{vh}p_{21}N_{h,2}}{\sum_{i=1}^np_{i1}N_{h,i}(\mu_v+\delta_1)} &\frac{ a_2\beta_{vh}p_{22}N_{h,2}}{\sum_{i=1}^np_{i2}N_{h,i}(\mu_v+\delta_2)} &\cdots &  \frac{ a_m\beta_{vh}p_{2m}N_{h,2}}{\sum_{i=1}^np_{im}N_{h,i}(\mu_v+\delta_m)}\\
\vdots & \vdots& \cdots & \vdots\\
\frac{a_1\beta_{vh}p_{n1}N_{h,n}}{\sum_{i=1}^np_{i1}N_{h,i}(\mu_v+\delta_1)} &  \frac{a_2\beta_{vh}p_{n2}N_{h,n}}{\sum_{i=1}^np_{i2}N_{h,i}(\mu_v+\delta_2)} & \cdots &  \frac{ a_m\beta_{hv}p_{nm}N_{h,n}}{\sum_{i=1}^np_{im}N_{h,i}(\mu_v+\delta_m)}
\end{array}\right)$$

Note that the matrices $M_{hv}$ and $M_{vh}$ are of $(m,n)$ and $(n,m)$ size respectively. The matrix $M_{vh}$ represents the new human cases due to infected mosquitoes  whereas $M_{hv}$ represents the new mosquito cases due to humans. In fact, the elements  of $M_{vh}$ and $M_{hv}$ have specific biological interpretations, for instance,
\begin{itemize}
\item For $i=1,\dots, n$ and $j=1,\dots, m$, we have $(m_{hv})_{ji}=\frac{ a_j\beta_{hv}p_{ij}\bar N_{v,j}}{(\mu_i+\gamma_i)\sum_{l=1}^np_{lj}N_{h,l}}$ represents the average number of secondary vector (of Patch $j$)  cases of infection produced by a single infected of group $i$ during his/her infectious period.
\item Similarly, for $i=1,\dots, n$ and $j=1,\dots, m$, we have $(m_{vh})_{ij} =\frac{ a_j\beta_{hv}p_{ij}N_{h,i}}{(\mu_v+\delta_j)\sum_{l=1}^np_{lj}N_{h,l}}$ represents the average number of secondary human (of group $i$) cases generated by an infected vector (of Patch $j$) during her infectious period. 

\item The overall number of new cases produced by an infected mosquitoes in Patch $j$ is the sum of the elements of the $j^{th}$ column of $M_{vh}$, that is,

\begin{eqnarray}
\sum_{i=1}^{n}(m_{vh})_{ij} &=&\sum_{i=1}^{n}\frac{ a_j\beta_{hv}p_{ij}N_{h,i}}{(\mu_v+\delta_j)\sum_{l=1}^np_{lj}N_{h,l}}\nonumber\\
&=&\frac{ a_j\beta_{hv}}{(\mu_v+\delta_j)\sum_{l=1}^np_{lj}N_{h,l}}\sum_{i=1}^{n}p_{ij}N_{h,i}\nonumber\\
&=&\frac{ a_j\beta_{hv}}{(\mu_v+\delta_j)}.\nonumber
\end{eqnarray}

\end{itemize}
The overall new vector cases generated by an infected human of group $i$  is the sum of the first column of $M_{hv}$, namely,

\begin{eqnarray}
\sum_{j=1}^{m}(m_{hv})_{ji}&=&\sum_{j=1}^{m}\frac{ a_j\beta_{hv}p_{ij}\bar N_{v,j}}{(\mu_i+\gamma_i)\sum_{l=1}^np_{lj}N_{h,l}}\nonumber\\
&=&\frac{\beta_{hv}}{(\mu_i+\gamma_i)}\sum_{j=1}^{m}\frac{ a_jp_{ij}\bar N_{v,j}}{\sum_{l=1}^np_{lj}N_{h,l}}.\nonumber
\end{eqnarray}

 The matrix $$\left(\begin{array}{cc}
0 & M_{vh}\\
M_{vh} & 0
\end{array}\right)$$
or equivalently the matrix $$\left(\begin{array}{cc}
0 & \mathcal N_{vh}\\
\mathcal N_{vh} & 0
\end{array}\right)$$
where $\mathcal N_{vh}=\textrm{diag}(N_h)\mathbb P\textrm{diag}(\mathbb P^tN_h)^{-1}\textrm{diag}(a)$ and $\mathcal N_{hv}=\textrm{diag}(a)\textrm{diag}(\mathbb P^t N_h)^{-1}\textrm{diag}(\bar N_v)\mathbb P^t$ is what authors in  \cite{iggidr2014dynamics} called host-vector network. 
  It is proven in \cite{iggidr2014dynamics} that the matrix $-FV^{-1}$ is irreducible if and only if $M_{hv}M_{vh}$ and $M_{vh}M_{hv}$
 are both irreducible. Notice that, even if $n=m$, the irreducibility of $\mathbb P$ is neither necessary nor sufficient to ensure the irreducibility of $M_{hv}M_{vh}$ and $M_{vh}M_{hv}$. See the subsection \ref{Subsection2P2G} for counter examples. \\
\begin{remark} \hfill

If we suppose that $n=m$ and $p_{ij}=0,\;\; \forall \{i,j\}\in\{1,2,\dots,n\}^2$ and $i\neq j$ then $\mathbb P=I_n$.  Hence, $M_{hv}$ and $M_{vh}$ are diagonal matrices and so is their product $M_{hv}M_{vh}$. Hence,
$$\mathcal R_0^2=\max\{(\mathcal R_0^1)^2,(\mathcal R_0^2)^2,\dots,(\mathcal R_0^n)^2\} $$
where $$(\mathcal R_0^i)^2=\frac{a_i^2\beta_{vh}\beta_{hv}\bar N_{v,i}}{\mu_v(\mu_i+\gamma_i)N_{h,i}}$$
For each $i$, the basic reproduction number $(\mathcal R_0^i)^2$ is the one derived from the classical Ross model.
\end{remark}

\begin{theorem}\label{TheoDFEGAS} Under the assumption that the host-vector network is irreducible, we have that
\begin{enumerate}
\item If $\mathcal R_0\leq1$, the DFE is globally asymptotically stable. 
\item If $\mathcal R_0>1$, the DFE is unstable.
\end{enumerate}
\end{theorem}

\begin{theorem}\label{TheoEEGAS}
Under the assumption that the vector-host network is irreducible, we have that if $\mathcal R_0>1$ then there exists a unique endemic equilibrium that is globally asymptotically stable.
\end{theorem}

Theorems \ref{TheoDFEGAS} and \ref{TheoEEGAS} can be obtained by using H. L. Smith's results \cite{smith1986cooperative}. Indeed, it is immediate that System (\ref{PatchMatrices}) is cooperative, strongly concave and its Jacobian is irreducible. Hence, the theorems can be obtained following Smith's results \cite{smith1986cooperative} (Theorem 3.1 and Corollary 3.2). A similar method using Hirsch's theorem \cite{Hirsch84} is outlined by Iggidr et \textit{al.} in \cite{iggidr2012global} for an $SIS$ metapopulation model.

A sharp threshold results of a multi-group vector-borne disease model has been obtained in \cite{iggidr2014dynamics} by using nicely crafted Lyapunov functions and elements of graph theory.
\section{Effects of heterogeneity}\label{sec:HeteroR0}
In this section, we take a closer look to the effect of heterogeneity on the basic reproduction number and  provide simple bounds for the basic reproduction number that may be useful in applications. We also compare the effects of \textit{patchiness} (the role of variable number of patches/environments) and \textit{groupness} (the role of variable hosts' groups) on the basic reproduction number. We denote $\mathcal R_0^2(m,n)$ the basic reproduction  number for $n$ groups and $m$ patches.

The  basic reproduction number is the spectral radius of $M_{vh}M_{hv}$. This matrix is supposed irreducible with entries are given by:

\begin{equation}\label{MvhMhv}
r_{ij}=\frac{\beta_{vh}\beta_{hv}N_{h,i}}{\mu_j+\gamma_j}\sum_{k=1}^{m}\frac{a_k^2p_{ik}p_{jk}N_{v,k}}{(\sum_{l=1}^{n}p_{lk}N_{h,l})^2(\mu_v+\delta_k)}\quad \forall i,j=1,\dots,n.
\end{equation}
$\mathbb R=(r_{ij})$ is an $n\times n$ matrix. The basic reproduction number is also the spectral radius of $M_{hv}M_{vh}$, a matrix with entries

\begin{equation}\label{MhvMvh}
\tilde r_{ij}=\frac{a_ia_j\beta_{hv}\beta_{vh}N_{v,i}}{(\mu_v+\delta_j)(\sum_{l=1}^np_{li}N_{h,l})(\sum_{l=1}^np_{lj}N_{h,l})}\sum_{k=1}^n\frac{p_{ki}p_{kj}N_{h,k}}{\mu_k+\gamma_k} \quad \forall i,j=1,\dots,m. \end{equation}

The next theorem collects a set of inequalities that identify lower and upper bounds for the basic reproduction number.

\begin{theorem}\label{TheoBounds} \hfill

\begin{enumerate}
\item $\displaystyle \min_{j=1,\dots,n}L_j \leq \mathcal R_0^2(n,m)\leq \max_{j=1,\dots,n}L_j$
where
$$L_j=\frac{\beta_{hv}\beta_{hv}}{\mu_j+\gamma_j}\sum_{k=1}^m\frac{a_k^2p_{jk}N_{v,k}}{(\sum_{l=1}^{n}p_{lk}N_{h,l})(\mu_v+\delta_k)}$$
\item $\displaystyle \min_{i=1,\dots,n}L_i^\sharp \leq \mathcal R_0^2(n,m)\leq \max_{i=1,\dots,n} L_i^\sharp$ where
$$L_i^\sharp=\sum_{k=1}^m\frac{a_k^2\beta_{hv}\beta_{hv}p_{ik}N_{h,i}N_{v,k}}{(\sum_{l=1}^{n}p_{lk}N_{h,l})^2(\mu_v+\delta_k)}\left(\sum_{j=1}^n\frac{p_{jk}}{\mu_j+\gamma_j}\right)$$
\item 
$\displaystyle \min_{j=1,\dots,m}L_j^\diamond \leq \mathcal R_0^2(n,m)\leq \max_{j=1,\dots,m}L_j^\diamond$ where
$$
L_j^\diamond=\frac{a_j\beta_{hv}\beta_{vh}}{(\mu_v+\delta_j)(\sum_{l=1}^np_{lj}N_{h,l})}\sum_{k=1}^{n}\frac{p_{kj}N_{h,k}}{\mu_k+\gamma_k} \left( \sum_{i=1}^m\frac{a_ip_{ki}N_{v,i}}{\sum_{l=1}^np_{li}N_{h,l}}   \right)
$$

\item  
$\displaystyle \min_{i=1,\dots,m}L_i^\dagger \leq \mathcal R_0^2(n,m)\leq \max_{i=1,\dots,m}L_i^\dagger$ where
$$
L_i^\dagger=\frac{a_i\beta_{hv}\beta_{vh}N_{v,i}}{\sum_{l=1}^np_{li}N_{h,l}}\sum_{k=1}^{n}\frac{p_{ki}N_{h,k}}{\mu_k+\gamma_k}\left( \sum_{j=1}^m\frac{a_jp_{kj}}{(\mu_v+\delta_j)\sum_{l=1}^np_{lj}N_{h,l}}   \right)
$$

\end{enumerate}
\end{theorem}
\begin{proof}
\begin{enumerate}
\item\label{TheoremR01}
Since $M_{vh}M_{hv}$ is a  nonnegative irreducible matrix, the basic reproduction number $\mathcal R_0^2=\rho(M_{vh}M_{hv})$ satisfy the Frobenius' inequality:
$$\min_j r_{j}(M_{vh}M_{hv})\leq \mathcal R_0^2(n,m)\leq\max_j r_{j}(M_{vh}M_{hv})$$
where $r_{j}(M_{vh}M_{hv})=\sum_{i=1}^{n}r_{ij}$ and $r_{ij}$ are given by (\ref{MvhMhv}). We have:

\begin{eqnarray}
r_{j}(M_{vh}M_{hv}) &=&\sum_{i=1}^{n}r_{ij}\nonumber\\
&=&\sum_{i=1}^{n}\frac{\beta_{vh}\beta_{hv}N_{h,i}}{\mu_j+\gamma_j}\sum_{k=1}^{m}\frac{a_k^2p_{ik}p_{jk}N_{v,k}}{(\sum_{l=1}^{n}p_{lk}N_{h,l})^2(\mu_v+\delta_k)}\nonumber\\
&=&\frac{\beta_{vh}\beta_{hv}}{\mu_j+\gamma_j}\sum_{k=1}^{m}\frac{a_k^2p_{jk}N_{v,k}}{(\sum_{l=1}^{n}p_{lk}N_{h,l})^2(\mu_v+\delta_k)}\sum_{i=1}^{n}p_{ik}N_{h,i}\nonumber\\
&=&\frac{\beta_{vh}\beta_{hv}}{\mu_j+\gamma_j}\sum_{k=1}^{m}\frac{a_k^2p_{jk}N_{v,k}}{(\sum_{l=1}^{n}p_{lk}N_{h,l})(\mu_v+\delta_k)}\nonumber\\
&:=& L_j \nonumber
\end{eqnarray}

\item This inequality is obtained in the same way as \ref{TheoremR01} this time, by summing over the columns of $M_{vh}M_{hv}$ and using Frobenius' inequality. 
\item \label{Th3.1.4} By considering the fact that $\mathcal R_0^2(n,m)$ is also the spectral radius of matrix $M_{hv}M_{vh}$, the Frobenius' inequality leads to
$$\min_j \tilde r_{j}(M_{hv}M_{vh})\leq \mathcal R_0^2(n,m)\leq\max_j \tilde r_{j}(M_{hv}M_{vh}),$$
where $\tilde r_{j}(M_{hv}M_{vh})=\sum_{i=1}^{m}\tilde r_{ij}$ and $\tilde r_{ij}$ are given by (\ref{MhvMvh}). It follows that,
\begin{eqnarray}
\tilde r_{j}(M_{hv}M_{vh}) &=&\sum_{i=1}^{m}\tilde r_{ij}\nonumber\\
&=&\sum_{i=1}^{m}\frac{a_ia_j\beta_{hv}\beta_{vh}N_{v,i}}{(\mu_v+\delta_j)(\sum_{l=1}^np_{li}N_{h,l})(\sum_{l=1}^np_{lj}N_{h,l})}\sum_{k=1}^n\frac{p_{ki}p_{kj}N_{h,k}}{\mu_k+\gamma_k} \nonumber\\
&=&\frac{a_j\beta_{hv}\beta_{vh}}{(\mu_v+\delta_j)(\sum_{l=1}^np_{lj}N_{h,l})}\sum_{i=1}^{m}\frac{a_ip_{ki}N_{v,i}}{(\sum_{l=1}^np_{li}N_{h,l})}\sum_{k=1}^n\frac{p_{kj}N_{h,k}}{\mu_k+\gamma_k} \nonumber\\
&:=& L_j^\diamond \nonumber
\end{eqnarray}
\item Let $r_i(M_{hv}M_{vh})$ denote the sum of the entries along the $i$-th row of $M_{hv}M_{vh}$. We have:
\begin{eqnarray}
\tilde r_{i}(M_{hv}M_{vh}) &=&\sum_{j=1}^{m}\tilde r_{ij}\nonumber\\
&=&\sum_{j=1}^{m}\frac{a_ia_j\beta_{hv}\beta_{vh}N_{v,i}}{(\mu_v+\delta_j)(\sum_{l=1}^np_{li}N_{h,l})(\sum_{l=1}^np_{lj}N_{h,l})}\sum_{k=1}^n\frac{p_{ki}p_{kj}N_{h,k}}{\mu_k+\gamma_k} \nonumber\\
&=&\frac{a_i\beta_{hv}\beta_{vh}N_{v,i}}{(\sum_{l=1}^np_{li}N_{h,l})}\sum_{j=1}^{m}\frac{a_jN_{v,i}}{(\mu_v+\delta_j)(\sum_{l=1}^np_{lj}N_{h,l})}\sum_{k=1}^n\frac{p_{ki}p_{kj}N_{h,k}}{\mu_k+\gamma_k} \nonumber\\
&:=& L_i^\dagger \nonumber
\end{eqnarray}
We deduce the inequality as in \ref{Th3.1.4}.
 
\end{enumerate}
\end{proof}
Note that the bounds $L_j$ and $L_j^\sharp$ can be interpreted biologically. $L_j$ is the sum of the products of the number of secondary cases of infections on mosquitoes (of Patch $k$, $k=1,\dots,m$) produced by infected host  of Group $j$ ($\frac{a_k\beta_{hv}\bar N_{v,k}}{\mu_j+\gamma_j}$) and secondary cases of infections on hosts (of Group $j$) produced by infected mosquitoes in Patch $k$ ($\frac{a_k\beta_{vh}}{\mu_v+\delta_k}\frac{p_{jk}}{\sum_{l=1}^np_{lk}N_{h,l}}$)  for  $k=1,\dots,m$.\\

Similarly, $L^\sharp$ is the sum of the product between the number of secondary cases produced by infected mosquitoes (of Patch $k$) on hosts, that is, $\displaystyle\frac{a_k\beta_{hv}p_{ik}N_{h,i}}{(\sum_{l=1}^{n}p_{lk}N_{h,l})(\mu_v+\delta_k)}$,
and the secondary mosquito cases of infection produced by infected hosts during their infectious period, that is, $\displaystyle\sum_{j=1}^n\frac{a_k}{\mu_j+\gamma_j}\frac{p_{jk}}{\sum_{l=1}^{n}p_{lk}N_{h,l}}$.

 \begin{theorem}
 If the residence time matrix is of rank one then an explicit expression of the basic reproduction number for the general system is given by
 $$\mathcal R_0^2(m,n)=\frac{\beta_{vh}\beta_{hv}}{(\sum_{l=1}^{n}p_{l}N_{h,l})^2}\sum_{k=1}^{m}\frac{a_k^2N_{v,k}}{(\mu_v+\delta_k)}\left( \sum_{i=1}^{n}\frac{p_{i}^2N_{h,i}}{\mu_i+\gamma_i}\right)$$
 \end{theorem}
\begin{proof} \hfill

Let us suppose that the residence time matrix $\mathbb P$ is of rank 1. There exist $x\in\mathbb R^n_+$ and $y\in\mathbb R^m_+$ such that $\mathbb P=xy^T$. Since $\mathbb P$ is stochastic, we can deduce that $x_i\sum_{j=1}^my_j=1$ for $i=1,2,\dots,n$. Therefore,  $\mathbb P$ could be written as $\mathbb P=\mathbbm{1}p^t$ where $p\in\mathbb R^m$ and $\sum_{i=1}^mp_i=1$. Hence, the matrices $M_{vh}M_{hv}$ and $M_{hv}M_{vh}$ are also of rank one. Therefore the trace of  $M_{vh}M_{hv}$ is only positive eigenvalue of  $M_{vh}M_{hv}$.  Hence, by using (\ref{MvhMhv}), we obtain:
\begin{eqnarray}
\mathcal R^2_0(m,n)&=&\sum_{i=1}^n\frac{\beta_{vh}\beta_{hv}N_{h,i}}{\mu_i+\gamma_i}\sum_{k=1}^{m}\frac{a_k^2p_{i}^2N_{v,k}}{(\sum_{l=1}^{n}p_{l}N_{h,l})^2(\mu_v+\delta_k)}\nonumber\\
&=&\frac{\beta_{vh}\beta_{hv}}{(\sum_{l=1}^{n}p_{l}N_{h,l})^2}\sum_{i=1}^n\frac{p_{i}^2N_{h,i}}{\mu_i+\gamma_i}\sum_{k=1}^{m}\frac{a_k^2N_{v,k}}{(\mu_v+\delta_k)}\nonumber\\
&=&\frac{\beta_{vh}\beta_{hv}}{(\sum_{l=1}^{n}p_{l}N_{h,l})^2}\sum_{k=1}^{m}\frac{a_k^2N_{v,k}}{(\mu_v+\delta_k)}\left( \sum_{i=1}^{n}\frac{p_{i}^2N_{h,i}}{\mu_i+\gamma_i}\right)\nonumber
\end{eqnarray}
\end{proof}
If we assume that ``virtual" dispersal does not induce any substantial change in the population of each patch, i.e $p_{i}N_{h,i}=N_{h,i}$ at any time and that $\mu_i=\mu$ and $\gamma_i=\gamma$ for all $i=1,2,\dots, n$, then we recover the result of  Dye and Hasibeder  \cite{dye1986population,hasibeder1988population}, namely,
$$\mathcal R^2_0(m,n)=\mathcal R^2_0(m,1),$$
where $\mathcal R^2_0(m,1)$ is the basic reproduction number corresponding of $m$ patches of vectors and a single host group. Similarly, if $\delta_j=\delta$ and $a_j=a$ for all $j=1,2,\dots,m$, we obtain $$\mathcal R^2_0(m,n)=\mathcal R^2_0(1,n),$$
where $\mathcal R^2_0(1,n)$ is the basic reproduction number corresponding of $n$ host groups and a single patch of vectors.\\

For $m$ patches and one group, the basic reproduction number is given by
$$\mathcal R_0^2(m,1)=\frac{\beta_{vh}\beta_{hv}}{(\mu+\gamma)N_{h}}\sum_{k=1}^{m}\frac{a_k^2p^2_{1k}N_{v,k}}{(\mu_v+\delta_k)}.$$

The basic reproduction number associated with  single group and single environment turns out to be the classical $\mathcal R_0^2$, that is, $\mathcal R_0^2(1,1)=\frac{a^2\beta_{vh}\beta_{hv}}{(\mu+\gamma)(\mu_v+\delta)}\frac{N_{v}}{N_{h}}$. We arrive at the following result:
\begin{lemma}\label{R0Vectors}
We have 
$$\mathcal R_0^2(m,1)\geq \mathcal R_0^2(1,1)$$
\end{lemma}
\begin{proof}\hfill
\begin{eqnarray}\label{R0M1R011}
\mathcal R_0^2(m,1) &=&\frac{\beta_{vh}\beta_{hv}}{(\mu_+\gamma_1)N_{h}}\sum_{k=1}^{m}\frac{a_k^2p^2_{1k}\bar N_{v,k}}{(\mu_v+\delta_k)}\\
&\geq& \frac{\beta_{vh}\beta_{hv}}{(\mu_+\gamma_1)N_{h}}\frac{a_k^2p_{11}^2\bar N_{v,1}}{(\mu_v+\delta_1)}:= \mathcal R_0^2(1,1)\nonumber
\end{eqnarray}
since $p_{11}=1$ for a single patch and single host.
\end{proof}

\begin{remark}
In Lemma \ref{R0Vectors}, we are comparing the basic reproduction number of the $m$ patches and 1 group case with the one of 1 patch and 1 group case. And so, the $p_{11}$ in the RHS of the inequality in the proof of Lemma \ref{R0Vectors}, is seen both as the $p_{11}$ of the single patch, single  group case and the $m$ patches, single group case. 
\end{remark}

Lemma \ref{R0Vectors} states that, for a single group, the presence of patch/environmental heterogeneity might increase the basic reproduction number. 
\section{Special cases and Simulations} \label{sec:Simulations}
In this section, we provide examples of cases where the number of patches and number of groups are either equal or non-equal, that is, we highlight in a limited way the role of \textit{patchiness} and \textit{groupness}. We start off by the case of two patches and two groups to showcase that even when the residence times matrix $\mathbb P$ is square, its irreducibility is neither necessary nor sufficient to ensure the irreducibility of the next generation matrix. This implies that the disease either dies out or persists in all patches and groups. We then consider the three patches and two groups for which the disease persists in all groups and patches in an attempt to see how the differential in residence times leads to a differential in the disease burden for hosts and vectors.   
\subsection{The two patches and two groups case}\label{Subsection2P2G}
 As stated at the derivation of the model (Section \ref{sec:Derivation}), this system could model either the case where there are two patches within which there are hosts and vectors or it could model the case where there are two groups of hosts interacting in two different patches. The basic reproduction number, for $n=2,\; m=2$ is $\rho(M_{vh}M_{hv})$, where
$$M_{hv}=\left(\begin{array}{cc}
\frac{a_1\beta_{hv}p_{11}\bar N_{v,1}}{(p_{11}N_{h,1}+p_{21}N_{h,2})(\mu_1+\gamma_1)} &\frac{ a_1\beta_{hv}p_{21}\bar N_{v,1}}{(p_{11}N_{h,1}+p_{21}N_{h,2})(\mu_2+\gamma_2)} \\
\frac{a_2\beta_{hv}p_{12}\bar N_{v,2}}{(p_{12}N_{h,1}+p_{22}N_{h,2})(\mu_1+\gamma_1)} &\frac{ a_2\beta_{hv}p_{22}\bar N_{v,2}}{(p_{12}N_{h,1}+p_{22}N_{h,1})(\mu_2+\gamma_2)}
\end{array}\right)$$
and
$$M_{vh}=\left(\begin{array}{ccc}
 \frac{a_1\beta_{vh} p_{11}N_{h,1}}{(p_{11}N_{h,1}+p_{21}N_{h,2})(\mu_v+\delta_1)} & \frac{a_2\beta_{vh} p_{12}N_{h,1}}{(p_{12}N_{h,1}+p_{22}N_{h,2})(\mu_v+\delta_2)} \\
 \frac{a_1\beta_{vh} p_{21}N_{h,2}}{(p_{11}N_{h,1}+p_{21}N_{h,2})(\mu_v+\delta_1)} & \frac{a_2\beta_{vh} p_{22}N_{h,2}}{(p_{12}N_{h,1}+p_{22}N_{h,2})(\mu_v+\delta_2)}\\
\end{array}\right)$$
We have that
{
%\tiny{
$$M_{vh}M_{hv}=\left(\begin{array}{cc}
m_{11} & m_{12} \\
m_{21}&m_{22}
\end{array}\right)$$

where $m_{11}=\frac{a_1^2\beta_{hv}\beta_{vh}p_{11}^2\bar N_{v,1}N_{h,1}}{(p_{11}N_{h,1}+p_{21}N_{h,2})^2(\mu_1+\gamma_1)(\mu_v+\delta_1)}+\frac{a_2^2\beta_{hv}\beta_{vh}p_{12}^2\bar N_{v,2}N_{h,1}}{(p_{12}N_{h,1}+p_{22}N_{h,2})^2(\mu_1+\gamma_1)(\mu_v+\delta_2)}$,
$$m_{12}=\frac{a_1^2\beta_{hv}\beta_{vh}p_{11}p_{21}\bar N_{v,1}N_{h,1}}{(p_{11}N_{h,1}+p_{21}N_{h,2})^2(\mu_2+\gamma_2)(\mu_v+\delta_1)} +\frac{a_2^2\beta_{hv}\beta_{vh}p_{12}p_{22}\bar N_{v,2}N_{h,1}}{(p_{12}N_{h,1}+p_{22}N_{h,2})^2(\mu_2+\gamma_2)(\mu_v+\delta_2)}, 
$$
$$m_{21}=\frac{a_1^2\beta_{hv}\beta_{vh}p_{11}p_{21}\bar N_{v,1}N_{h,2}}{(p_{11}N_{h,1}+p_{21}N_{h,2})^2(\mu_1+\gamma_1)(\mu_v+\delta_1)} +\frac{a_2^2\beta_{hv}\beta_{vh}p_{12}p_{22}\bar N_{v,2}N_{h,2}}{(p_{12}N_{h,1}+p_{22}N_{h,2})^2(\mu_1+\gamma_1)(\mu_v+\delta_2)}, 
$$
and 
$$m_{22}=\frac{a_1^2\beta_{hv}\beta_{vh}p_{21}^2\bar N_{v,1}N_{h,2}}{(p_{11}N_{h,1}+p_{21}N_{h,2})^2(\mu_2+\gamma_2)(\mu_v+\delta_1)}+\frac{a_2^2\beta_{hv}\beta_{vh}p_{22}^2\bar N_{v,2}N_{h,2}}{(p_{12}N_{h,1}+p_{22}N_{h,2})^2(\mu_2+\gamma_2)(\mu_v+\delta_2)}$$
We observe, that even for the case $n=m$, the irreducibility of $\mathbb P$ is nor necessary not sufficient to ensure the irreducibility of $M_{vh}M_{hv}$ and $M_{hv}M_{vh}$. Indeed, if $p_{12}=0$, $p_{21}>0$ and $p_{22}>0$, the residence time matrix is given by
$$\left(\begin{array}{cc}1& 0\\ p_{21} & p_{22}\end{array}\right)
$$ is reducible whereas
$$M_{vh}M_{hv}=\left(\begin{array}{cc}
\frac{a_1^2\beta_{hv}\beta_{vh}p_{11}^2\bar N_{v,1}N_{h,1}}{(p_{11}N_{h,1}+p_{21}N_{h,2})^2(\mu_1+\gamma_1)(\mu_v+\delta_1)}& \frac{a_1^2\beta_{hv}\beta_{vh}p_{11}p_{21}\bar N_{v,1}N_{h,1}}{(N_{h,1}+p_{21}N_{h,2})^2(\mu_2+\gamma_2)(\mu_v+\delta_1)}  \\
\frac{a_1^2\beta_{hv}\beta_{vh}p_{21}\bar N_{v,1}N_{h,2}}{(N_{h,1}+p_{21}N_{h,2})^2(\mu_1+\gamma_1)(\mu_v+\delta_1)}  &\frac{a_1^2\beta_{hv}\beta_{vh}p_{21}^2\bar N_{v,1}N_{h,2}}{(N_{h,1}+p_{21}N_{h,2})^2(\mu_2+\gamma_2)(\mu_v+\delta_1)}+\frac{a_2^2\beta_{hv}\beta_{vh}\bar N_{v,2}}{N_{h,2}(\mu_2+\gamma_2)(\mu_v+\delta_2)}
\end{array}\right)$$
is irreducible. Similarly, the residence times matrix $$\mathbb P=\left(\begin{array}{cc}0& 1\\ 1 & 0\end{array}\right)$$ is irreducible while the non-diagonal entries of $M_{vh}M_{hv}$ are equal to zero, that is, $m_{12}=m_{21}=0$. Hence, $M_{vh}M_{hv}$ is not irreducible.
If the matrices $M_{hv}M_{vh}$ and $M_{vh}M_{hv}$ are not both irreducible, we may obtain boundary equilibria for which the disease dies out in some hosts' groups and vectors' patches while persisting in others. See Fig \ref{IhReducible} and \ref{IvReducible} for instance.
%%%%%
%%%%%
\subsection{The three patches and two groups case} 
As an illustrative example, we consider System (\ref{PatchGenF}) for the case $n=2$ groups and $m=3$ patches. The basic reproduction number is the spectral radius of 
\begin{equation}
   -FV^{-1}= \begin{pmatrix}%{ccc}
        \begin{array}{cc}0 & 0\\ 0 & 0\end{array} & M_{vh} \\% \cline{1-2}
     M_{hv}&    \begin{array}{ccc}0 & 0 & 0\\ 0 & 0 & 0\end{array} 
    \end{pmatrix}
\end{equation}

where 
$$M_{hv}=\left(\begin{array}{cc}
\frac{a_1\beta_{hv}p_{11}\bar N_{v,1}}{(p_{11}N_{h,1}+p_{21}N_{h,2})(\mu_1+\gamma_1)} &\frac{ a_1\beta_{hv}p_{21}\bar N_{v,1}}{(p_{11}N_{h,1}+p_{21}N_{h,2})(\mu_2+\gamma_2)} \\
\frac{a_2\beta_{hv}p_{12}\bar N_{v,2}}{(p_{12}N_{h,1}+p_{22}N_{h,2})(\mu_1+\gamma_1)} &\frac{ a_2\beta_{hv}p_{22}\bar N_{v,2}}{(p_{12}N_{h,1}+p_{22}N_{h,1})(\mu_2+\gamma_2)}\\
\frac{a_3\beta_{hv}p_{13}\bar N_{v,3}}{(p_{13}N_{h,1}+p_{23}N_{h,2})(\mu_1+\gamma_1)} &\frac{ a_3\beta_{hv}p_{23}\bar N_{v,3}}{(p_{13}N_{h,1}+p_{23}N_{h,1})(\mu_2+\gamma_2)} 
\end{array}\right)$$
and
$$M_{vh}=\left(\begin{array}{ccc}
 \frac{a_1\beta_{vh} p_{11}N_{h,1}}{(p_{11}N_{h,1}+p_{21}N_{h,2})(\mu_v+\delta_1)} & \frac{a_2\beta_{vh} p_{12}N_{h,1}}{(p_{12}N_{h,1}+p_{22}N_{h,2})(\mu_v+\delta_2)} & \frac{a_3\beta_{vh} p_{13}N_{h,1}}{(p_{13}N_{h,1}+p_{23}N_{h,2})(\mu_v+\delta_3)}\\
 \frac{a_1\beta_{vh} p_{21}N_{h,2}}{(p_{11}N_{h,1}+p_{21}N_{h,2})(\mu_v+\delta_1)} & \frac{a_2\beta_{vh} p_{22}N_{h,2}}{(p_{12}N_{h,1}+p_{22}N_{h,2})(\mu_v+\delta_2)}&\frac{a_3\beta_{vh} p_{23}N_{h,2}}{(p_{13}N_{h,1}+p_{23}N_{h,2})(\mu_v+\delta_3)}\\
\end{array}\right)$$

For purposes of simulations, we use the following baseline parameters with the ranges given in parentheses.
$$\beta_{hv}=0.5 (0.001-0.54),\;\;\; \beta_{vh}=0.41 (0.3-0.9)\;\; \frac{1}{\mu_v}=20 (10-30) \;\textrm{days}, \frac{1}{\mu_1}=75\times 365 \;\textrm{days}$$
$$a_1=0.5\;\textrm{day}^{-1},\;\; a_2=0.4\;\textrm{day}^{-1},\;\; a_3=0.3\;\textrm{day}^{-1}$$
 $$\frac{1}{\mu_2}=73\times 365\;\textrm{day}s \;\textrm{days},\; \frac{1}{\gamma_1}=7 \;\textrm{days},\; \frac{1}{\gamma_2}=6 \;\textrm{days},\; \delta_1=0.001\; \textrm{day}^{-1},\;\delta_2=0.01,\;\delta_3=0.08\;\textrm{day}^{-1}.$$
The values of $\beta_{hv}$, $\beta_{vh}$ and $\mu_v$ are taken from \cite{chitnis2013modelling}. The host populations and the recruitments of vectors for the 3 patches are taken as
 $$N_{h,1}=4000,\;\; N_{h,2}=4500,\;\; \Lambda_{v,1}=1000,\;\; \Lambda_{v,2}=1000,\;\; \Lambda_{v,3}=950.$$
Unless otherwise stated, we fix $p_{13}=0.1$ and $p_{23}=0.2$, carrying out System (\ref{PatchGenF}) simulations that focus on the effects of non-fixed residence times matrix entries on the prevalence of hosts and vectors.  

\begin{figure}[ht]
\centering
 \subfigure[The level of prevalence of host of group 1 seems to decrease as $p_{12}$ increases (and hence $p_{11}$ decreases).]{
   \includegraphics[scale =.35]{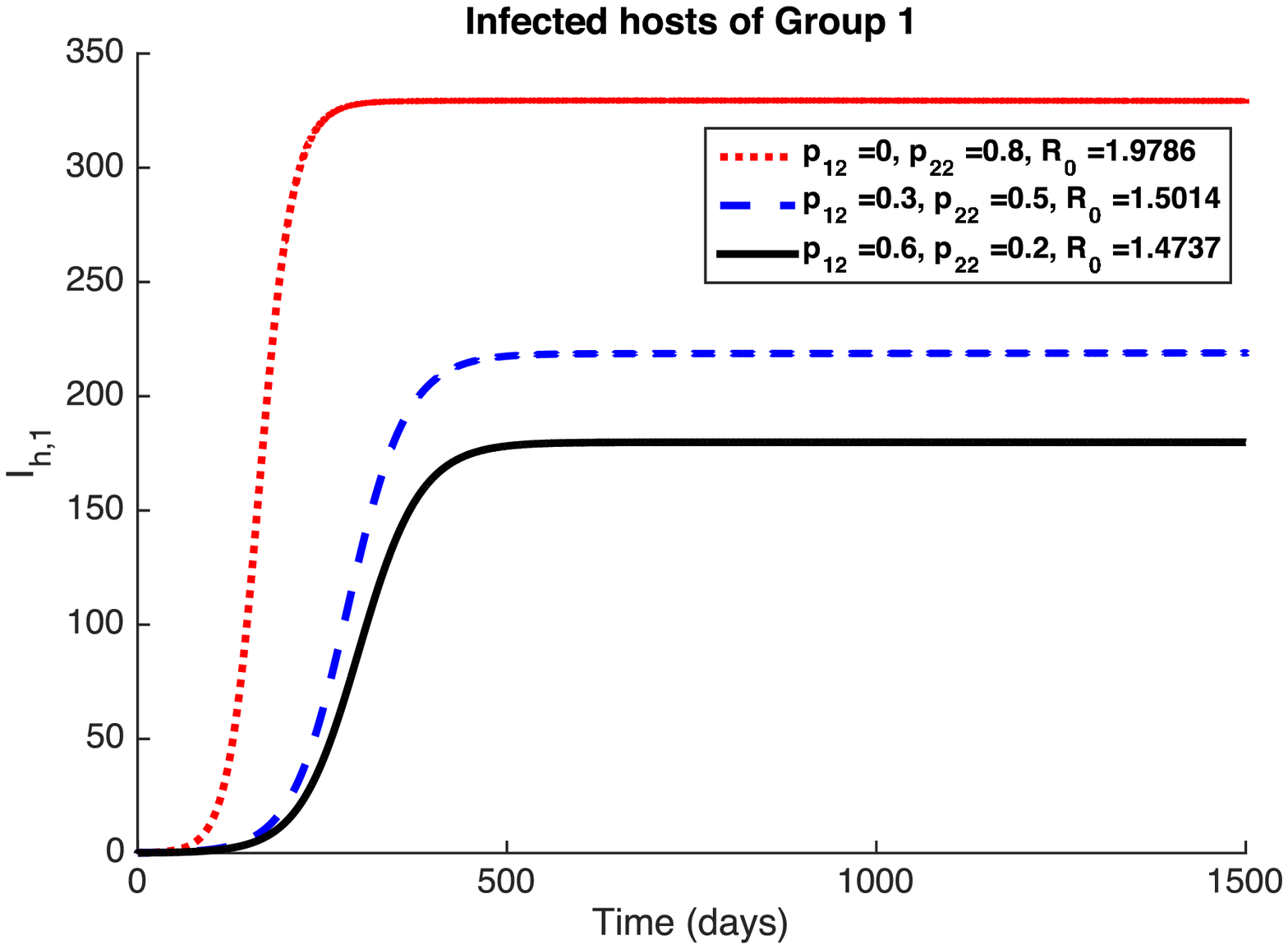}
\label{Ih1Normal2}}
\hspace{1mm}
 \subfigure[The level of prevalence of host of group 2 seems to decrease with respect to $p_{22}$.]{
   \includegraphics[scale =.35] {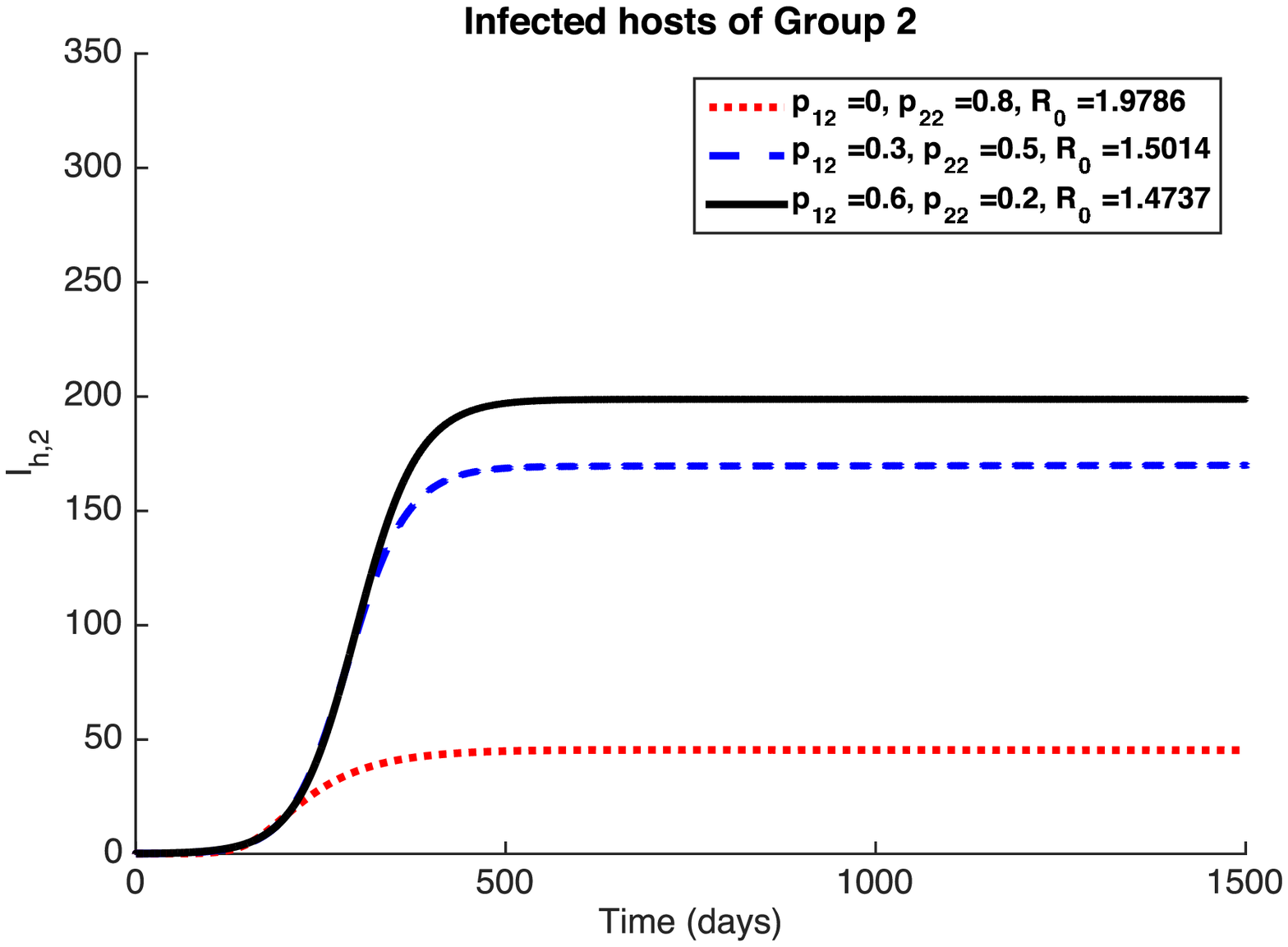}
\label{Ih2Normal2}}
\caption{Dynamics of $I_{h,1}$ and $I_{h,2}$ for different values of $p_{ij}$.} \label{fig:twofigsIh}
\end{figure}

Figure \ref{fig:twofigsIh} displays the dynamics of infected hosts of Group 1 (Fig \ref{Ih1Normal2}) and Group 2 (Fig \ref{Ih2Normal2}). The level of endemicity of individuals of Group 1 seems to decrease as proportion of time in Patch 2 ($p_{12}$) increases; probably because as $p_{12}$ increases, $p_{11}$ decreases. In other words, individuals of Group 1 spend more time in the less riskier Patch 2 ($a_2=0.4$) than in the riskier Patch 1 ($a_1=0.5$).  We see that the less time that  individuals spend in riskier environment the less likely that they will become infected, as one would expect.

In Figure \ref{Ih2Normal2}, the level of endemicity of hosts in Group 2 seems to decrease as $p_{22}$ increases or equivalently as $p_{21}$ decreases ( $p_{21}+p_{22}+p_{23}=1$ and $p_{23}$ is fixed).  It is so because individuals are increasing their residence time in Patch 2 ($a_2=0.4$) rather than in Patch 1.

\begin{figure}[ht]
\centering
 \subfigure[The level of prevalence of vectors of Patch 1 is decreasing as $p_{12}$ increases and $p_{22}$ increases.]{
   \includegraphics[scale =.35] {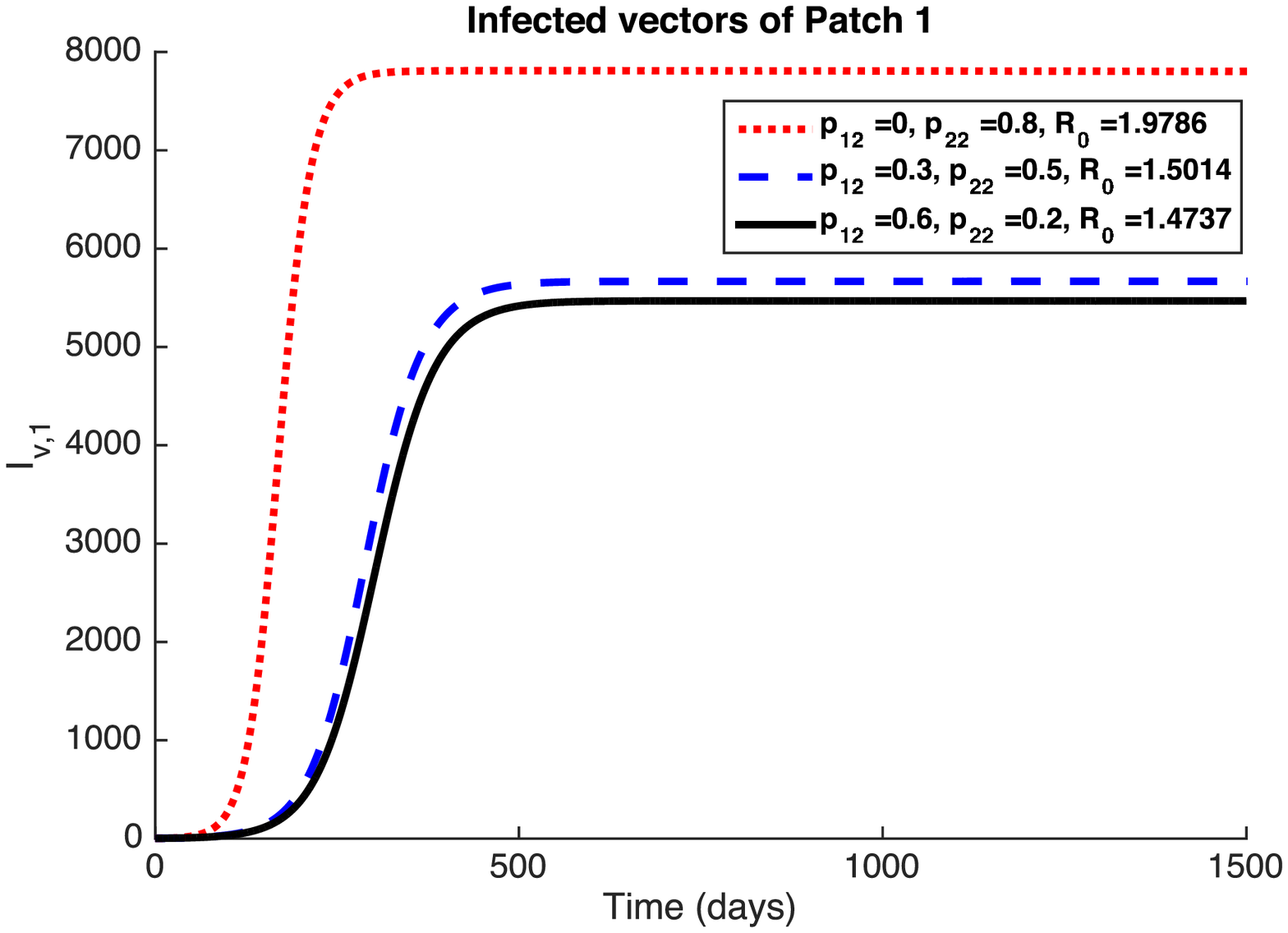}
\label{Iv1}}
\hspace{1mm}
 \subfigure[The level of prevalence of host of group 2]{
   \includegraphics[scale =.35] {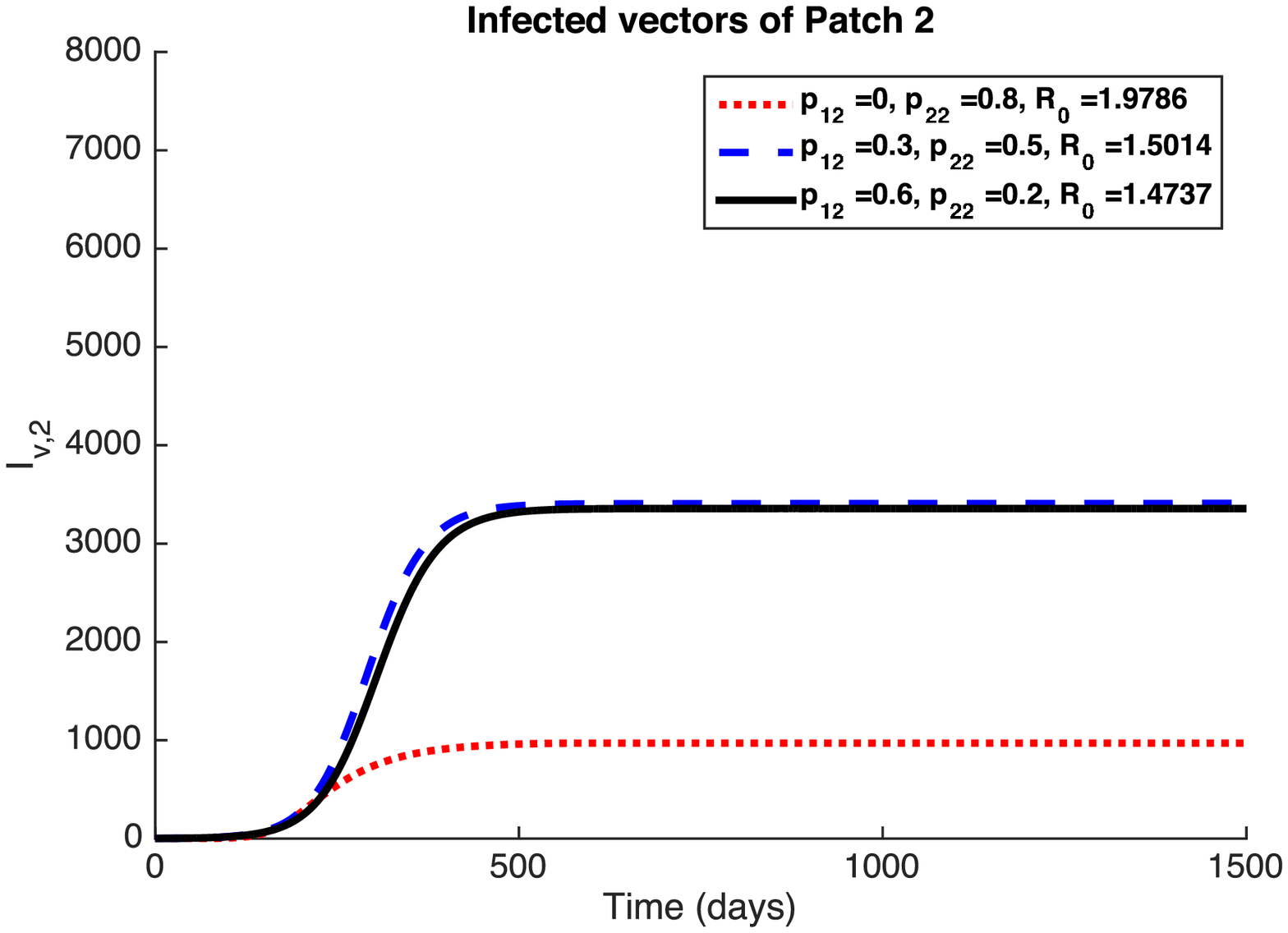}
\label{Iv2}}
\caption{Dynamics of $I_{v,1}$ and $I_{v,2}$ for different values of $p_{ij}$.} \label{fig:twofigsIv1Iv2}
\end{figure}
Figures \ref{fig:twofigsIv1Iv2} and  \ref{fig:Iv3Normal} offer an overview on how the dynamics of vectors change as the proportion of time  that individuals of Group 1 and Group 2 spend in environments 1, 2 and 3 varies. For the selected residence times matrix entries, the prevalence of vector in environment 1 (see Fig \ref{Iv1}) is at its highest if $p_{12}=0$ and $p_{22}=0.8$. With this configuration,  $p_{11}=0.9$ and Patch 1 has the highest \textit{effective} population size. Moreover, Patch 1 has the highest biting rate, leading to high level of vector infections in that patch.

 Though $\Lambda_{v,1}=\Lambda_{v,2}$, the prevalence of vectors Patch 2 (Fig \ref{Iv2}) is lower than of Patch 1 (Fig \ref{Iv1}), regardless of the combination of the chosen residence times entries. This is because the \textit{effective} population of Patch 2 is less than  the \textit{effective} population of Patch 1 for all the three selected residence time configurations and also because $a_1>a_2$.
 \begin{figure}[ht!]
\begin{center}
\includegraphics[scale=0.5]{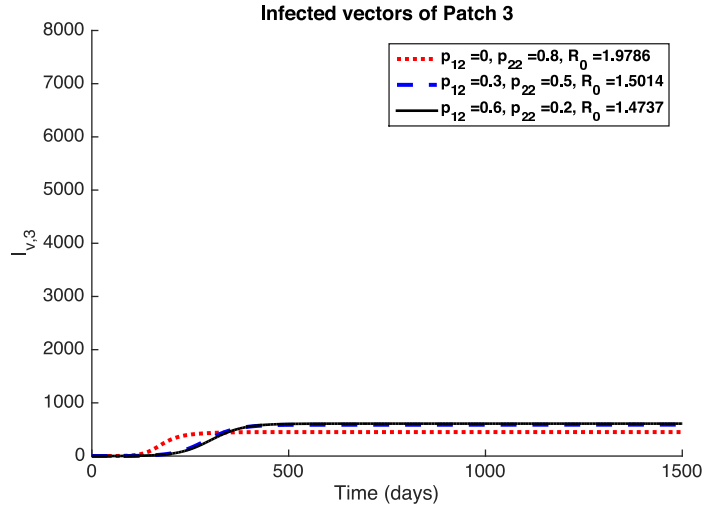}
\caption{Dynamics of vectors in Patch 3 ($I_{v,3}$) for different values of $p_{ij}$.}
\label{fig:Iv3Normal} 
\end{center}
\end{figure}

Figure  \ref{fig:Iv3Normal} represent the dynamics of the vectors in Patch 3. The number of infected vector in this patch is much less when compared to the number of infected in Patch 1 and 2. Again, the \textit{effective} population size of Patch 3, with $p_{13}=0.1$ and $p_{23}=0.2$, is much less when compared to those in Patch 1 and Patch 2.  Additionally, we also have by assumption that $a_3<a_2<a_1$.\\

For the vectors' prevalence, we obtain similar results as in Fig (\ref{fig:twofigsIv1Iv2})  and Fig (\ref{fig:Iv3Normal}) even if when biting rates are equal in all the three patches. This last comment  highlights the role of the \textit{effective} population per patch.
\begin{remark}
For all the selected combination of the residence times matrix entries in Fig \ref{fig:twofigsIh}-\ref{fig:twofigsIv1Iv2}-\ref{fig:Iv3Normal}, the matrices $M_{vh}M_{hv}$ and $M_{hv}M_{vh}$ are irreducible and $\mathcal R_0>1$ and so the curves of the infected hosts and vectors in those figures are reaching an endemic equilibrium level in accordance with Theorem \ref{TheoEEGAS}.
\end{remark}

\begin{figure}[ht]
\centering
 \subfigure[Dynamics of infected hosts of Group 1 and Group 2.]{
   \includegraphics[scale =.40]{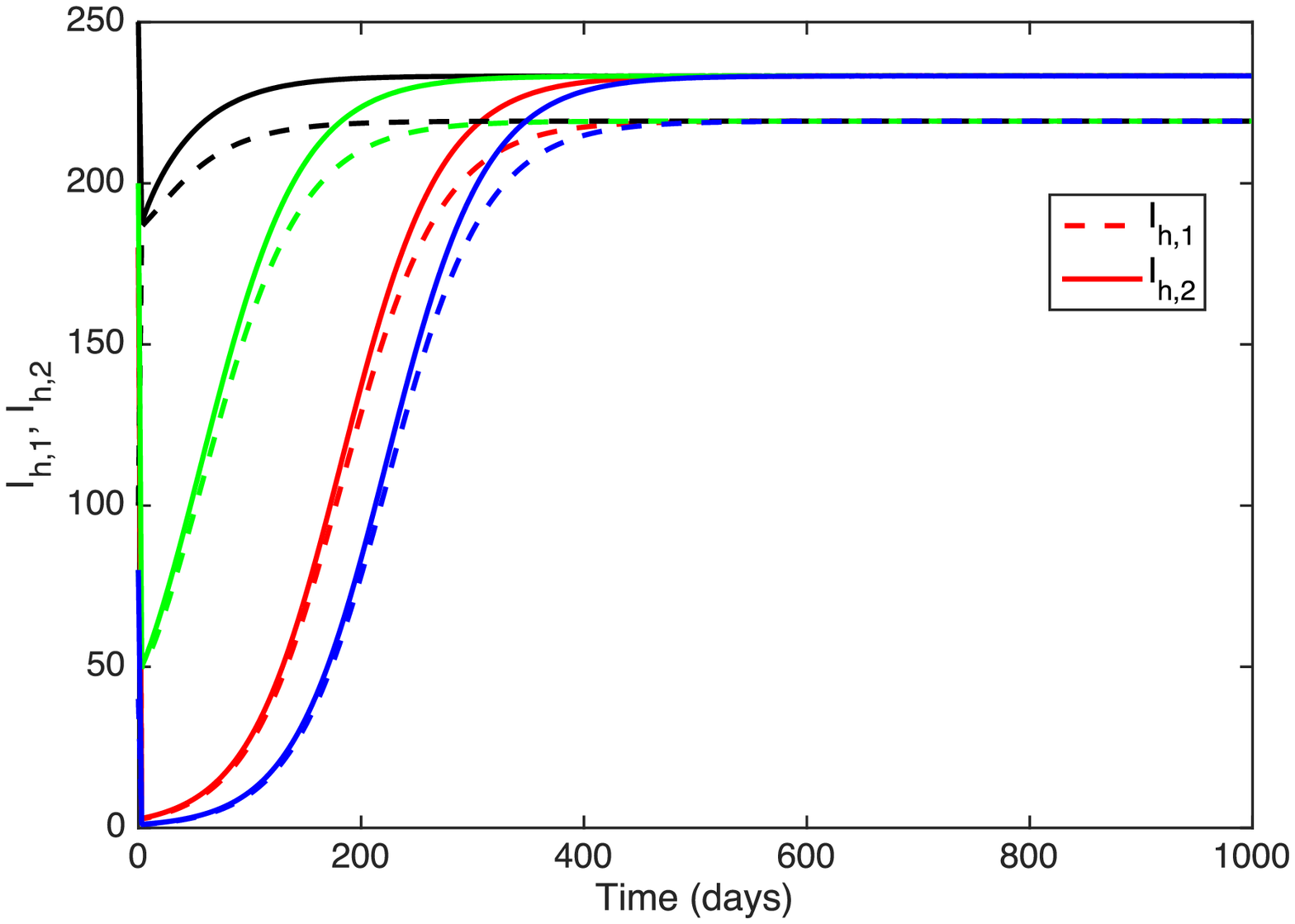}
\label{GASEEIh}}
\hspace{1mm}
 \subfigure[Dynamics of infected vectors of Patch 1 and Patch 2.]{
   \includegraphics[scale =.40]{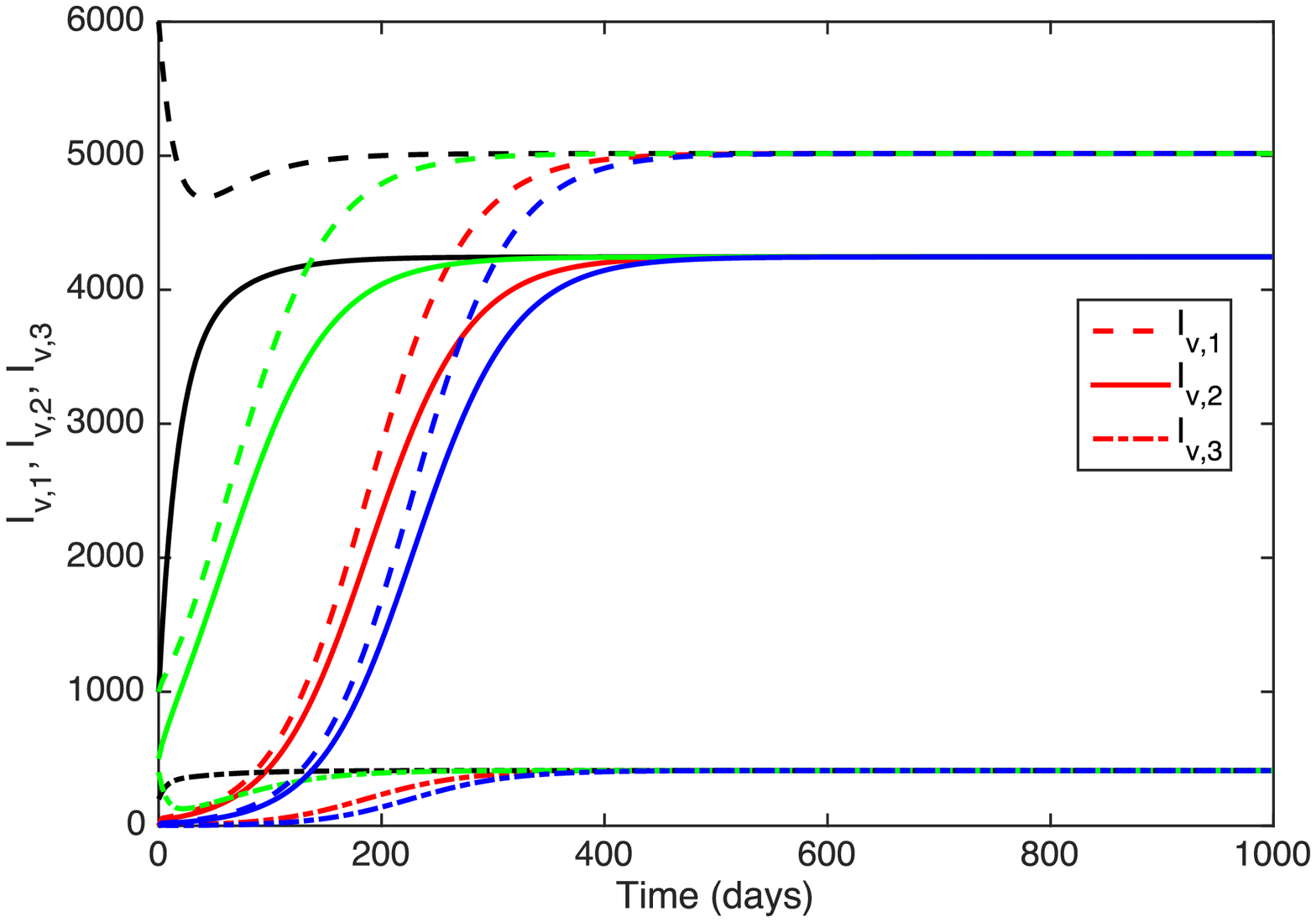}
\label{GASEEIv}}
\caption{Trajectories of System (\ref{PatchGenF}), with $n=2$ groups and $m=3$ patches with 4 different initial conditions. The trajectories are converging toward a unique interior endemic equilibrium.} \label{fig:twofigsGASEE}
\end{figure}
In Figure \ref{fig:twofigsGASEE} and \ref{fig:twofigsGASDFE}, we consider the case where the residence time matrix is fixed as follows:
$$\mathbb P=\left(\begin{array}{ccc}
0.4 & 0.3 & 0.3\\ 0.4 & 0.4 & 0.2
  \end{array}\right)$$
  In that case, the matrices $M_{hv}M_{vh}$ and $M_{vh}M_{hv}$ are both irreducible. We sketch the trajectories of System (\ref{PatchGenF}) for different initial conditions. Since, with these values of parameters and residence times matrix, the basic reproduction number is $\mathcal R_0(3,2)=1.4771$,  and so the result of Theorem  \ref{TheoEEGAS} should hold and  Figure \ref{fig:twofigsGASEE} confirms that. More precisely,  Figure \ref{GASEEIh} shows that the trajectories of infected individuals of Group 1 and Group 2 converge to the same interior endemic equilibrium for four different initial conditions, namely $IC_1=[I_{h,1}(0)=180,\; I_{h,2}(0)=180,\; I_{v,1}(0)=0,\; I_{v,2}(0)=0,\; I_{v,3}(0)=0]$ (solid red for $I_{h,1}$ and dashed red for $I_{h,2}$ in Fig \ref{GASEEIh}.), $IC_2=[I_{h,1}(0)=100,\; I_{h,2}(0)=250,\; I_{v,1}(0)=6000,\; I_{v,2}(0)=1000,\; I_{v,3}(0)=200]$ (solid black for $I_{h,1}$ and dotted black for $I_{h,2}$ in \ref{GASEEIh}.), $IC_3=[I_{h,1}(0)=80,\; I_{h,2}(0)=200,\; I_{v,1}(0)=1000,\; I_{v,2}(0)=500,\; I_{v,3}(0)=400]$ (solid green for $I_{h,1}$ and dotted green for $I_{h,2}$ in \ref{GASEEIh}.)  and $IC_4=[I_{h,1}(0)=40,\; I_{h,2}(0)=80,\; I_{v,1}(0)=1,\; I_{v,2}(0)=2,\; I_{v,3}(0)=3]$ (solid blue for $I_{h,1}$ and dotted blue for $I_{h,2}$ in \ref{GASEEIh}.).\\
  Similarly, Figure \ref{GASEEIv} displays the trajectories of infected vectors in the three considered environments. For all the above-mentioned initial conditions, these trajectories converge to their interior endemic equilibrium level. 
  
  %%%%%GAS DFE

\begin{figure}[ht]
\centering
 \subfigure[Dynamics of infected hosts of Group 1 and Group 2.]{
   \includegraphics[scale =.40]{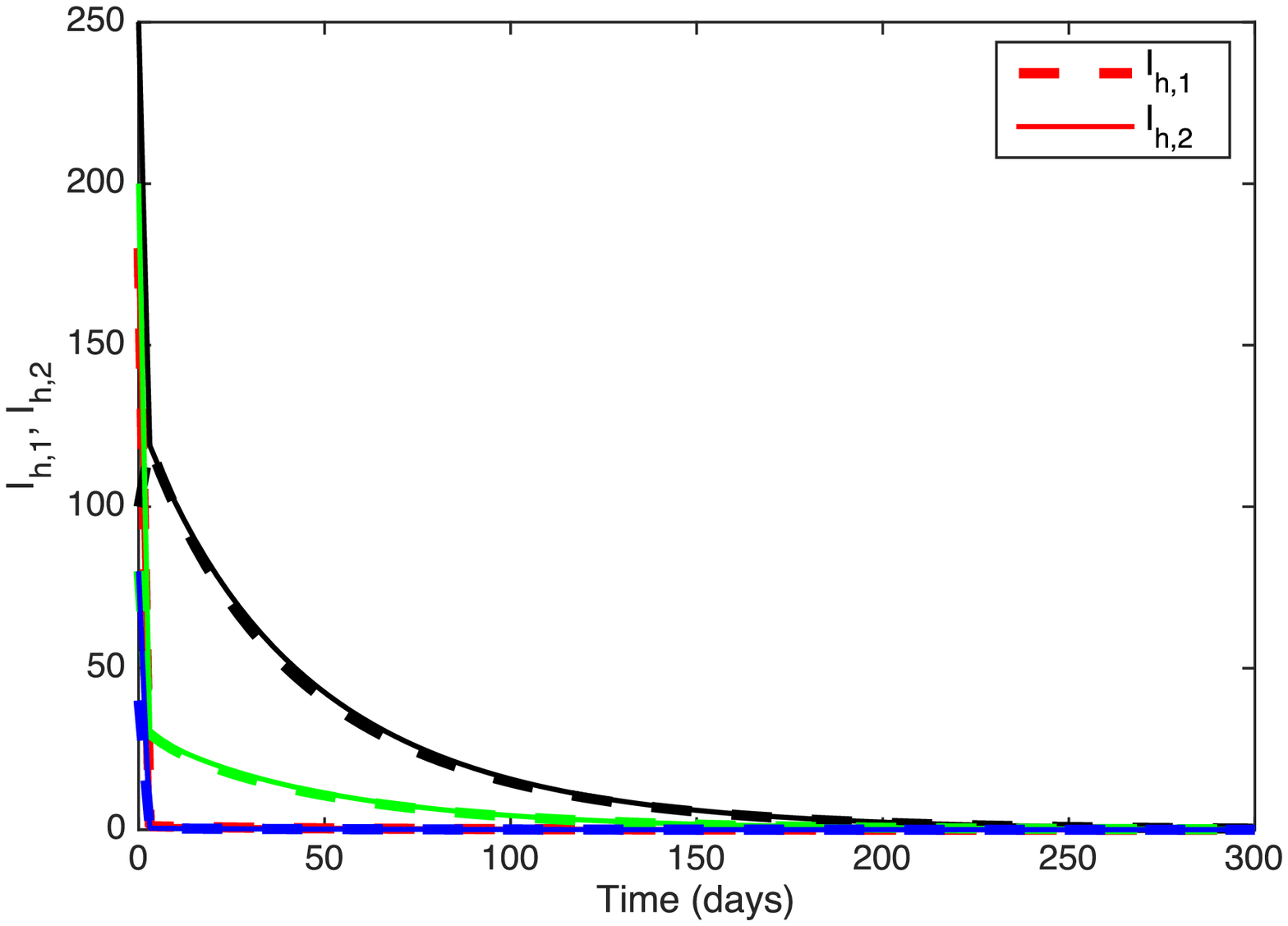}
\label{IhGASDFE}}
\hspace{1mm}
 \subfigure[Dynamics of infected vectors of Patch 1 and Patch 2.]{
   \includegraphics[scale =.40]{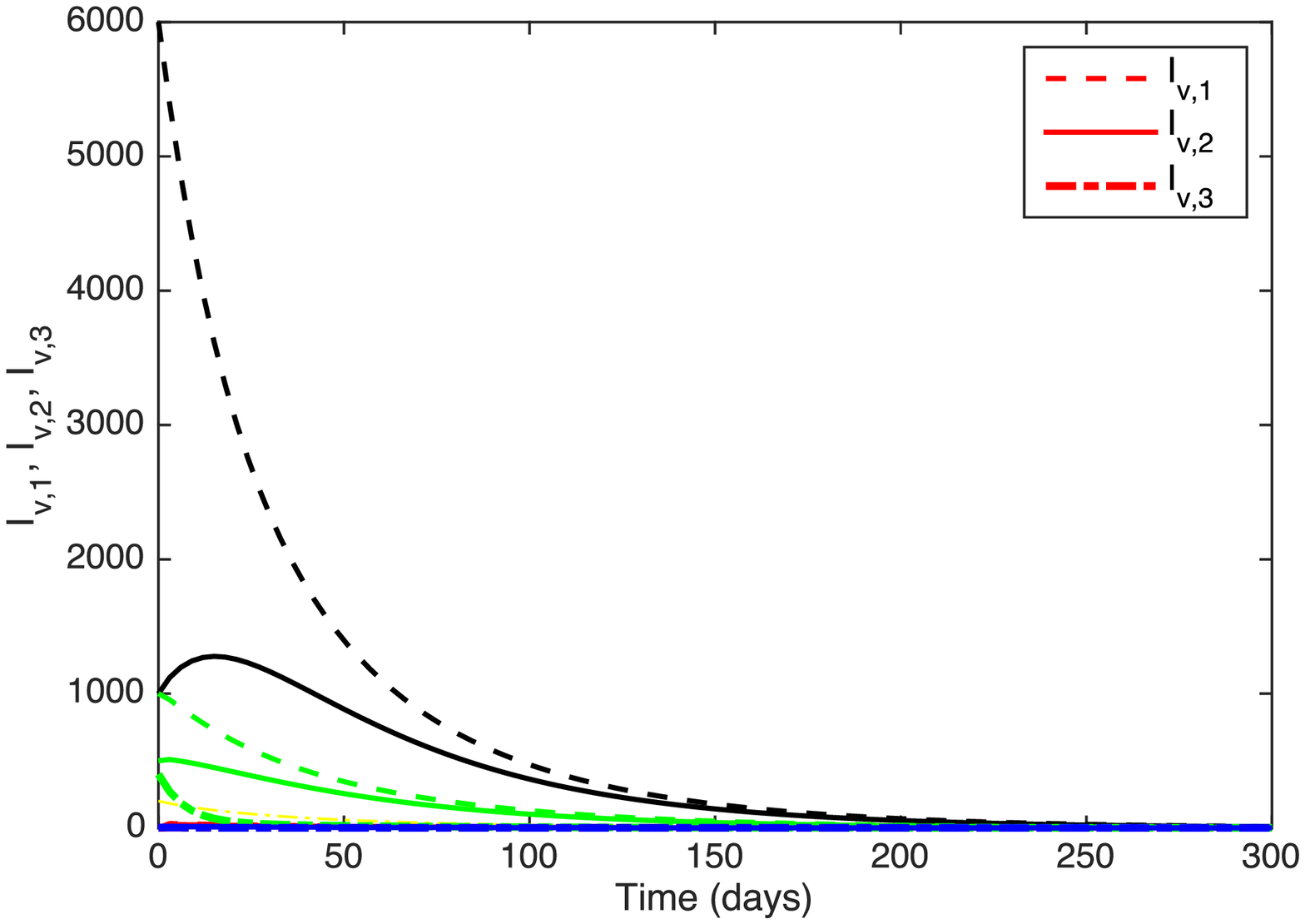}
\label{IvGASDFE}}
\caption{Trajectories of System (\ref{PatchGenF}), with $n=2$ groups and $m=3$ Patches with 4 different initial conditions. With $\beta_{hv}=0.2$ and $\beta_{vh}=0.4$, we have have $\mathcal R_0(3,2)=0.6353$ and the trajectories are converging toward the disease free equilibrium.} \label{fig:twofigsGASDFE}
\end{figure}

Figure \ref{fig:twofigsGASDFE} sketches the case where the values of all the parameters are the same as above but where $\beta_{hv}=0.2$ and $\beta_{vh}=0.4$. In this case, the basic reproduction number is $\mathcal R_0=0.6353$ which is less than one. As we can see in Fig \ref{IhGASDFE} and \ref{IvGASDFE}, the trajectories of infected hosts of the two groups and the infected vectors of the three patches are converging to zero for the above four initial conditions. This suggests that the DFE is globally asymptotically stable and confirms the result of Theorem \ref{TheoDFEGAS}.

Now, we consider the case where the configuration of the Group-Patch network is not irreducible. If we assume that $p_{12}=p_{21}=p_{23}=0$, the residence times matrix becomes
$$\mathbb P=\left(\begin{array}{ccc}
0.7 & 0& 0.3\\ 0 & 1 & 0
  \end{array}\right).$$
  This imply that $M_{hv}=\left(\begin{array}{cc}
\frac{a_1\beta_{vh}\bar N_{v,1}}{N_{h,1}(\mu_1+\gamma_1)} & 0\\
 0 & \frac{a_2\beta_{vh}\bar N_{v,2}}{N_{h,2}(\mu_2+\gamma_2)}\\ \frac{a_3\beta_{vh}\bar N_{v,3}}{N_{h,1}(\mu_1+\gamma_1)} & 0
 \end{array}\right)$ and $M_{vh}=\left(\begin{array}{ccc}
\frac{a_1\beta_{hv}}{\mu_v+\delta_1} & 0 & \frac{a_3\beta_{hv}}{\mu_v+\delta_3} \\
 0 & \frac{a_2\beta_{hv}}{\mu_v+\delta_2} & 0
 \end{array}\right)$. Hence, the matrices $M_{vh}M_{hv}$ and $M_{hv}M_{vh}$ are not both irreducible and hence Theorem  \ref{TheoEEGAS} does not hold, as shown in Fig \ref{fig:twofigsIReducible}, where a boundary equilibrium appears. In this case, members of Group 2 spend all their time in Patch 2 and hence are isolated from the rest of groups and patches. The basic reproduction number of Group 2 in Patch 2 is $(\mathcal R_0(2,2))^2= \frac{a_2^2\beta_{vh}\beta_{hv}N_{v,2}}{(\mu_v+\delta_2)(\mu_2+\gamma_2)N_{h,2}}=0.8$. The diseases dies out from the hosts of Group 2 (see \ref{IhReducible}, solid curves) and vectors of Patch 2 (\ref{IvReducible}, solid curves). Members of Group 1  are connected to Patch 1 and Patch 3. Hence, the corresponding basic reproduction number is $(\mathcal R_0^{1,{1,3}})^2= \frac{a_1^2\beta_{vh}\beta_{hv}N_{v,21}}{(\mu_v+\delta_1)(\mu_1+\gamma_1)N_{h,1}}+\frac{a_3^2\beta_{vh}\beta_{hv}N_{v,3}}{(\mu_v+\delta_3)(\mu_1+\gamma_1)N_{h,1}}=1.8549$. Hence, the disease persists among the members of Group 1(see \ref{IhReducible}, dashed curves) and vectors of Patch 1 and 3 (see \ref{IvReducible}, dashed curves). This case offers a glimpse on how disease dynamics when some groups are strongly connected to some environments while other groups are isolated. 
 \begin{figure}[ht]
\centering
 \subfigure[Dynamics of infected hosts of Group 1 and Group 2.]{
   \includegraphics[scale =.40]{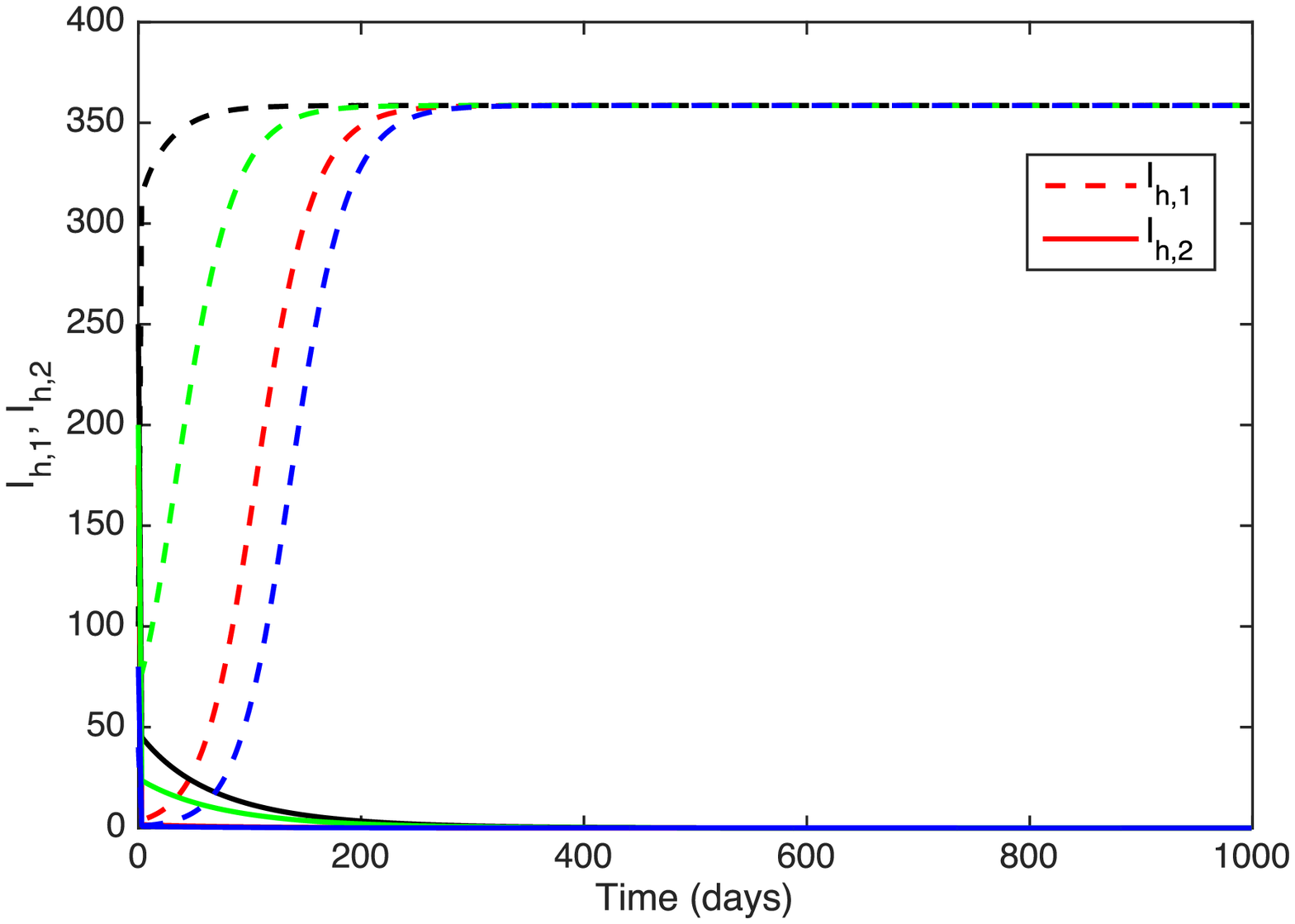}
\label{IhReducible}}
\hspace{1mm}
 \subfigure[Dynamics of infected vectors of Patch 1, Patch 2 and Patch 3.]{
   \includegraphics[scale =.40]{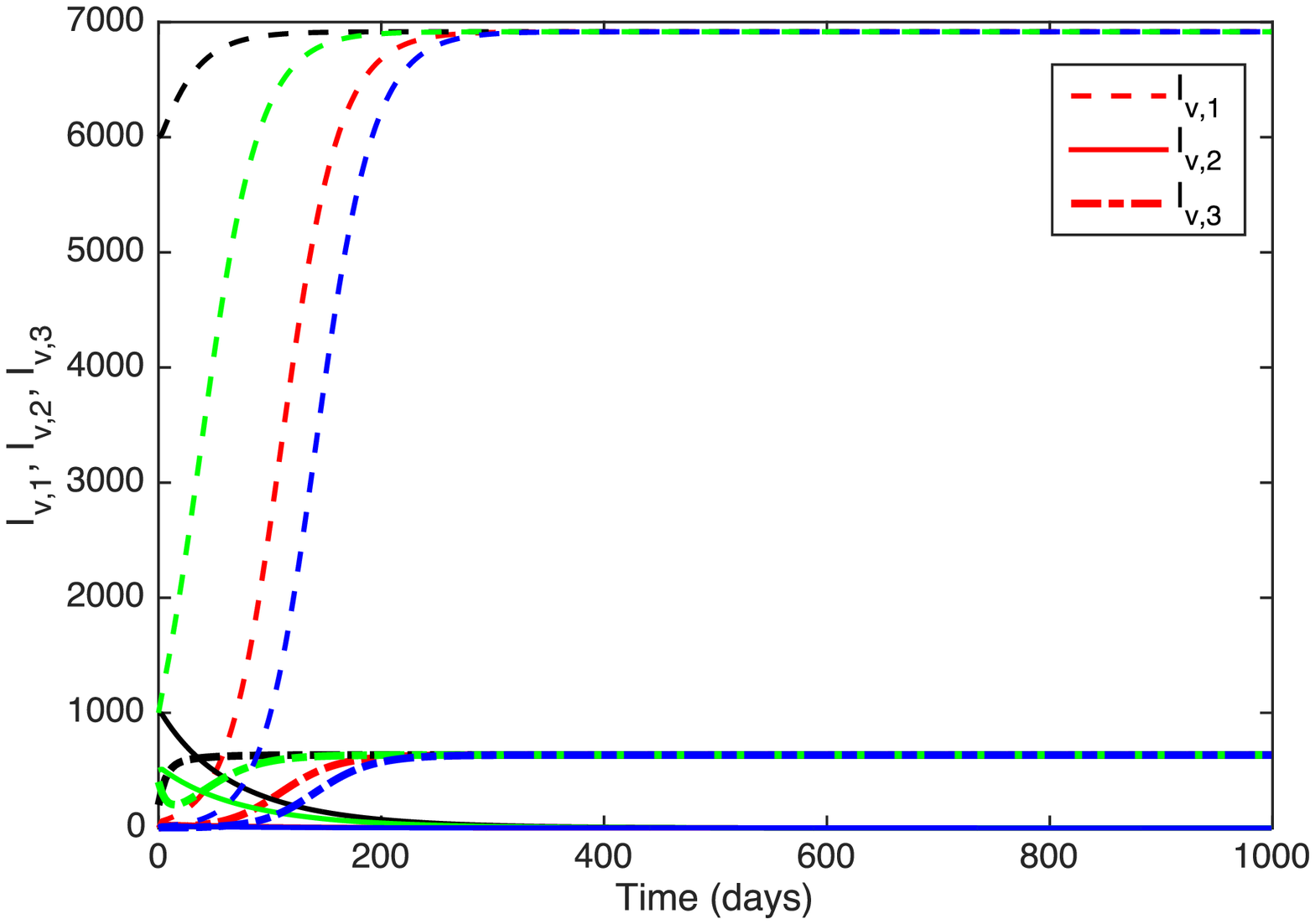}
\label{IvReducible}}
\caption{Trajectories of System (\ref{PatchGenF}), with $n=2$ groups and $m=3$ patches with 4 different initial conditions. The disease dies out for the host of Group 2 whereas it persists for those of Group 1. Similarly, the disease dies out for the vector of Patch 2 but persists for the vectors of Patch 1 and 3.} \label{fig:twofigsIReducible}
\end{figure}

\section{Conclusion and discussion}\label{sec:Conclusion}

Modeling vector-borne interactions have often been based  on well-mixed models that make it difficult to address effectively the role of host mobility on vector borne disease dynamics. Here, we consider a Lagrangian framework where  hosts' dispersal is modeled via the proportion of time that individuals  spend in different environments. In the process, we are forced to account, for time variations in \textit{effective} population size within each patch/environment. The kind of natural adjustment that can significantly alter the quantitative and qualitative dynamics of vector borne dynamics in geographically heterogeneous system; here within spatial scales that make it possible to neglect vector mobility.

And so, we consider  a general  $SIS$  framework to account for the host dynamics and an $SI$ framework to account for the vector dynamics. The transmission terms must make adjustments to account for the \textit{effective} population size generated by the residence time matrix. This is handled via the use of a modified  frequency-dependent incidence model that accounts for the \textit{effective} density of infected hosts within each patch at any time. We compute the basic reproduction number  $\mathcal R_0^2(\mathbb P,m,n)$ for the general host-vector model and prove that  the disease free equilibrium is globally asymptotically stable (GAS) if $\mathcal R_0^2(\mathbb P,m,n)\leq1$. We also show that there exists a unique interior endemic equilibrium that is GAS whenever  $\mathcal R_0^2(\mathbb P,m,n)>1$ in the irreducible case, that is, when the hosts' groups and vector patches are strongly connected. When irreducibility does not hold,   the existence of boundary equilibria is identified. In addition, we provide explicit expression for the basic reproduction number whenever the residence time matrix $\mathbb P$ is of rank one. Finally, we briefly explore  the role of variability in the number of patches and groups on the basic reproduction number, $\mathcal R_0^2(\mathbb P,m,n)$ and in the process, bounds for $\mathcal R_0^2(\mathbb P,m,n)$ are identified.

Our results generalize those of   \cite{cosner2009effects,lee2015role,rodriguez2001models} since our models account for the time-dependant \textit{effective} patch population size.
The approach we considered here includes the case where the hosts' structure is defined by residency ( see Subsection \ref{Subsection2P2G}) as well as the case when the hosts' structure is defined by groups or classes that are independent from the spatially explicit patches.

In short, the contributions of this manuscript are primarily tied to the \textit{effective} population size, a function of the mobility matrix $\mathbb P=(p_{ij})_{\substack{
      1 \leq i \leq n, \\
      1 \leq j \leq m}}$, where the $p_{ij}$ denotes the proportion of time the host of group $i$ (or  the member of a well defined class $i$) spends in environment $j$. We explicitly study the role of the matrix $\mathbb P$ on $\mathcal R_0$ and connected the dynamics to the reducibility and irreducibility structure of the system. Theorems were established and examples provided on the role of \textit{pachiness} and  \textit{groupness} on the disease dynamics. 
      
      \subsection*{Acknowledgements}
{\small{We are grateful to the handling editor and two anonymous reviewers for their helpful comments and suggestions which led to an improvement of this paper.  C.C.C is supported in part grant \#1R01GM100471-01 from the National Institute of General Medical Sciences (NIGMS) at the National Institutes of Health. The contents of this manuscript are solely the responsibility of the authors and do not necessarily represent the official views of DHS or NIGMS. The funders had no role in study design, data collection and analysis, decision to publish, or preparation of the manuscript.}}

% \nocite{*}
%\bibliography{VectBorneRTBIB22222}
%\bibliographystyle{siam}

\end{document}